\documentclass[a4paper]{article}
\pdfoutput=1

\usepackage{amssymb}
\usepackage{amsmath}
\usepackage{amsfonts}
\usepackage{amsthm}
\usepackage{multirow}
\usepackage{wrapfig}
\usepackage{url}
\usepackage{authblk}

\usepackage{todonotes}

\newlength \reqlen
\usepackage{tikz}
\usetikzlibrary{calc,decorations.pathmorphing}

\usepackage{listings, multicol,mathtools,tocloft}

\lstdefinelanguage{algo}{%
  morekeywords={function,algorithm,push,pop,top,for,all,and,or,if,then,else,repeat,until,while,do,report,return,such,that,each,add,call,exit,let,Input,
  Output, procedure}
}

\newcommand{\Aa}{\mathcal{A}}
\newcommand{\Oo}{\mathcal{O}}
\newcommand{\Pp}{\mathcal{P}}

\newcommand{\Gg}{\mathcal{G}}
\newcommand{\Ga}{\mathcal{G}_\mathcal{A}}

\renewcommand{\d}{\delta}

\newcommand{\s}{\sigma}

\newcommand{\D}{\Delta}

\newcommand{\true}{\operatorname{true}}
\newcommand{\false}{\operatorname{false}}
\newcommand{\Land}{\bigwedge}
\newcommand{\Lor}{\bigvee}
\newcommand{\imp}{\Rightarrow}
\newcommand{\ttt}{\operatorname{true}}
\newcommand{\fff}{\operatorname{false}}

\newcommand{\incl}{\subseteq}

\newcommand{\Rpos}{\mathbb{R}_{\ge 0}}
\newcommand{\Nat}{\mathbb{N}}
\newcommand{\Int}{\mathbb{Z}}

\newcommand{\NP}{\operatorname{NP}}
\newcommand{\PSPACE}{\operatorname{PSPACE}}

\theoremstyle{plain}
\newtheorem{theorem}{Theorem}
\newtheorem{lemma}{Lemma}
\newtheorem{proposition}[theorem]{Proposition}
\newtheorem{definition}[theorem]{Definition}
\newcommand{\Repeat}[2]{\noindent{$\blacktriangleright$~~\textsc{\textbf{#1~\ref{#2}.\
      }}}}

\newcommand{\xra}{\xrightarrow}
\newcommand{\dra}{\dashrightarrow}
\newcommand{\vali}{\mathbf{0}}
\newcommand{\ZG}{ZG}
\newcommand{\tto}{\Rightarrow}
\newcommand{\abs}{\mathfrak{a}}
\newcommand{\simmA}{\fleq}

\newcommand{\fleq}{\preccurlyeq}
\newcommand{\alu}{\abs_{\scriptscriptstyle \fleq LU }}
\newcommand{\lleq}{\mathrel{\triangleleft}}

\newcommand{\Ggbar}{\overline{\Gg}}
\newcommand{\act}[1]{\stackrel{#1}{\longrightarrow}}
\newcommand\sem[1]{{[\![ #1 ]\!]}}
\newcommand{\es}{\emptyset}

\newcommand{\reg}[1]{[#1]^{\scriptscriptstyle d}_{\scriptscriptstyle
M}}
\newcommand{\MinSum}{\operatorname{MinSum}}
\newcommand{\regeq}{\simeq}
\newcommand{\gv}{G_{\langle v \rangle}^*}
\newcommand{\gu}{G_{\langle u \rangle}^*}
\newcommand{\tight}{\multimapboth}
\newcommand{\redarrow}{\leadsto}
\newcommand{\Prop}{\operatorname{Prop}}
\newcommand{\Vars}{\operatorname{Vars}}
\newcommand{\0}{\mathbf{0}}
\newcommand{\1}{\mathbf{1}}
\newcommand{\n}{\mathbf{n}}
\newcommand{\red}{{\operatorname{red}}}

\newcommand{\blue}{{\operatorname{blue}}}

\newcommand{\clauses}{\operatorname{Clauses}}
\newcommand{\var}{\operatorname{Var}}

\newcommand{\yellowedge}{\dra}

\newcommand{\lu}{\fleq_{\scriptscriptstyle LU}}
\newcommand{\lud}{\fleq_{\scriptscriptstyle LU}^{\scriptscriptstyle
d}}
\newcommand{\md}{\fleq_{\scriptscriptstyle M}^{\scriptscriptstyle
d}}

\newcommand{\gvlu}{G_{\langle v \rangle}^{\scriptscriptstyle LU}}
\newcommand{\gvm}{G_{\langle v \rangle}^{\scriptscriptstyle M}}
\newcommand{\gvplu}{G_{\langle v' \rangle}^{\scriptscriptstyle LU}}
\newcommand{\lureg}[1]{\langle #1 \rangle_{\scriptscriptstyle
LU}}
\newcommand{\alud}{\mathfrak{a}_{\scriptscriptstyle
\fleq LU}^{\scriptscriptstyle d}}

\newcommand{\down}[1]{{\downarrow}{#1}}

\newcommand{\ExtraLUp}{\operatorname{Extra}_{\scriptscriptstyle LU}^+}

\DeclareSymbolFont{symbolsC}{U}{txsyc}{m}{n}
\SetSymbolFont{symbolsC}{bold}{U}{txsyc}{bx}{n}
\DeclareFontSubstitution{U}{txsyc}{m}{n}
\DeclareMathSymbol{\multimapboth}{\mathrel}{symbolsC}{19}

\makeatletter
\newcommand{\xdashrightarrow}[2][]{\ext@arrow 0359\rightarrowfill@@{#1}{#2}}
\newcommand{\xdashleftarrow}[2][]{\ext@arrow 3095\leftarrowfill@@{#1}{#2}}
\newcommand{\xdashleftrightarrow}[2][]{\ext@arrow 3359\leftrightarrowfill@@{#1}{#2}}
\def\rightarrowfill@@{\arrowfill@@\relax\relbar\rightarrow}
\def\leftarrowfill@@{\arrowfill@@\leftarrow\relbar\relax}
\def\leftrightarrowfill@@{\arrowfill@@\leftarrow\relbar\rightarrow}
\def\arrowfill@@#1#2#3#4{%
  $\m@th\thickmuskip0mu\medmuskip\thickmuskip\thinmuskip\thickmuskip
   \relax#4#1
   \xleaders\hbox{$#4#2$}\hfill
   #3$%
}
\makeatother

\newcommand{\tightone}{\xrightarrow{1}}
\newcommand{\tighttwo}{\xrightarrow{2}}

\usepackage{rotating}
\newcounter{sarrow}

\begin{document}

\title{\textbf{Reachability in timed automata with diagonal constraints}}

\author[1]{Paul Gastin}
\author[2]{Sayan Mukherjee}
\author[2]{B Srivathsan}
\affil[1]{\small LSV, ENS Paris-Saclay, CNRS, France and UMI 2000 ReLaX}
\affil[2]{\small Chennai Mathematical Institute, India and
    UMI 2000 ReLaX}
\date{}

\maketitle

\begin{abstract}
We consider the reachability problem for timed automata having
diagonal constraints (like $x - y < 5$) as guards in transitions.  The
best algorithms for timed automata proceed  by enumerating reachable
sets of its configurations, stored in the form of a data structure
called ``zones''. Simulation relations between zones are essential to ensure
termination and efficiency. The algorithm employs a simulation test
of the form $Z \fleq Z'$ which ascertains that zone $Z$ does not
 reach more states than zone $Z'$, and hence further enumeration from
 $Z$ is not necessary. No effective simulations are
known for timed automata containing diagonal constraints as guards. In this paper,
we propose a simulation relation $\lud$ for timed automata with
diagonal constraints. On the negative side, we show that deciding $Z
\not \lud$  is NP-complete. On the positive side, we
identify a witness for $Z \not \lud Z'$ and propose an
algorithm to decide the existence of such a witness using an SMT 
solver. The shape of the witness reveals that the simulation test is
likely to be efficient in practice.  

\end{abstract}

\section{Introduction}
\label{sec:introduction}

Timed automata~\cite{Alur:TCS:1994} are models of real-time
systems. They are finite automata equipped with real valued variables
called \emph{clocks}. These clocks can be used to constrain the time
difference between events: for instance when event $a$ occurs a clock
$x$ can be set to $0$ in the transition reading $a$, and when an event
$b$ occurs, the transition reading $b$ can check if $x \le 4$. These
constraints on clocks are called \emph{guards} and clocks which are
made $0$ in a transition are said to be \emph{reset} in the
transition. Guards of the form $x - y > 5$ are called \emph{diagonal
  constraints}. They are convenient for checking conditions about
events in the past: when an event $c$ occurs, we want to check that
between events $a$, $b$ which occurred previously (in the said order),
the time gap is at least $5$. One can then reset a clock $x$ at $a$,
$y$ at $b$ and check for $x - y > 5$ at $c$. It is known that such
diagonal constraints do not add to the expressive power: each timed
automaton can be converted into an equivalent automaton with no
diagonal guards, that is, a \emph{diagonal-free}
automaton~\cite{Berard:1998:FundInf}. However, this conversion leads
to an exponential blowup in the number of states, which is unavoidable
in general~\cite{Bouyer:2005:Conciseness}.

State reachability is a basic question in timed automata
verification. The problem is to decide if there exists a run of the
automaton from the initial state to a given accepting state. This is
known to be PSPACE-complete~\cite{Alur:TCS:1994}. In practice, the
best algorithms for reachability proceed by a forward analysis of the
automaton: starting from its initial state, enumerate reachable sets
of its configurations stored in the form of a data structure called
\emph{zones}. Zones are conjunctions of difference constraints (like
$x - y < 6~\land~ w > 4$) which can be efficiently represented and
manipulated using Difference Bound
Matrices~\cite{Dill:1990:DBM}. \emph{Abstractions} of zones are
necessary for termination and efficiency of this enumeration. These
abstractions are functions with a finite range mapping each set of
configurations to a bigger set.  For diagonal free timed automata
various implementable abstraction functions are known
\cite{Behrmann:STTT:2006,Herbreteau:IandC:2016}. For timed automata
with diagonal constraints, no such abstraction functions are known and
such a forward analysis method does not work. A naive method would be
to analyze the equivalent diagonal free automaton, but then this
introduces a (systematic)
blowup. 

Abstractions of zones can be used in two ways during the forward
analysis: explicitly or implicitly. In the explicit case: each time a
new zone $Z$ appears, the abstraction function $\abs$ is applied on it
and $\abs(Z)$ is stored. Further enumeration starts from
$\abs(Z)$. For this explicit method to work, $\abs(Z)$ needs an
efficient representation. Hence only abstractions where $\abs(Z)$ is
also a zone (also called \emph{convex abstractions}) are
used. $\ExtraLUp$\cite{Behrmann:STTT:2006} is the best known convex
abstraction for diagonal free automata and is implemented in the
state-of-the-art tool UPPAAL~\cite{Bengtsson:2004:UPPAAL}. In the
implicit case, zones are not extrapolated and are stored as they
are. Each time a new zone $Z$ appears, it is checked if there exists
an already visited zone $Z'$ such that $Z \incl \abs(Z')$. Intuitively
this means that zone $Z$ cannot see more states than $Z'$ and hence
the enumeration at $Z$ can stop. Given that $\abs$ has finite range,
the computation terminates. Since abstractions of zones are not stored
explicitly, there is no restriction for $\abs$ to result in a zone,
but an efficient inclusion test $Z \incl \abs(Z')$ is necessary as
this test is performed each time a new zone appears. For diagonal-free
automata, the best known abstraction is $\alu$ and it subsumes
$\ExtraLUp$. The inclusion test $Z \incl \alu(Z')$ can be done in
$\Oo(|X|^2)$ where $X$ is the number of
clocks~\cite{Herbreteau:IandC:2016}. In both cases - explicit or
implicit - it is important to have an abstraction that transforms
zones into as big sets as possible, so that the enumeration can
terminate with fewer zone visits.

In this paper, we are interested in the implicit method for timed
automata with diagonal constraints. Since the abstractions that are
usually used are based on simulation relations, the inclusion test $Z
\incl \abs(Z')$ boils down to a simulation test $Z \fleq Z'$ between
zones. In particular, the $\alu$ abstraction is based on a simulation
relation $\lu$~\cite{Behrmann:STTT:2006}. We choose to view the use of
implicit abstractions as simulations between zones. From the next
section, we refrain from using abstractions and present them as
simulations instead. We propose a simulation $\lud$, an extension to
$\lu$ that is sound for diagonal constraints. Contrary to the diagonal
free case, we show that the simulation test $Z \not \lud Z'$ is
NP-complete. But on the positive side, we give a characterization of a
witness for the fact that $Z \not \lud Z'$ and encode the existence of
such a witness as the satisfiability of a formula in linear
arithmetic. This gives an algorithm for $Z \not \lud Z'$. The shape of
the witness shows that in practice the number of potential candidates
would be few and the simulation test is likely to be efficient. We
have implemented our algorithm in a prototype
tool.
Preliminary experiments demonstrate that the number of zones
enumerated using $\lud$ simulation drastically reduces compared to the
number of zones obtained by doing the diagonal free conversion
followed by a forward analysis using $\lu$. This simulation relation
$\lud$ and the associated simulation test also open the door for
extending optimizations studied for diagonal free
automata~\cite{Herbreteau:2013:CAV}, to the case of diagonal
constraints; and also extending analysis of priced timed automata with
diagonal constraints~\cite{Bouyer::2016:CAV,Reynier:toolreport}.

\textbf{Related work:} Convex abstractions used for diagonal free
timed automata had been in use also for diagonal constraints in tools
like UPPAAL and KRONOS~\cite{KRONOS}. It was shown in
\cite{Bouyer:2004:forwardanalysis} that this is incorrect: there are
automata with diagonal constraints for which using $\ExtraLUp$ will
give a yes answer to the reachability problem, whereas the accepting
state is not actually reachable in the automaton. This is because the
extra valuations added during the computation enable guards which were
originally not enabled in the automaton, leading to spurious
executions.  A non convex abstraction for diagonal constraints appears
in \cite{Bouyer:2004:forwardanalysis}, but the corresponding inclusion
test is not known. Current algorithm for diagonal constraints proceeds
by an abstraction refinement
method~\cite{Bouyer:2005:diagonal-refinement}.

\textbf{Organization of the paper:} Section~\ref{sec:preliminaries}
gives the preliminary definitions. In
Section~\ref{sec:new-simul-relat}, we propose a simulation relation
$\lud$ between zones and observe some of its properties. Section
\ref{sec:algor-zone-incl-LU} gives an algorithm for $Z \not \lud Z'$
via reduction to an SMT formula. Section \ref{sec:hardness} shows that
$Z \not \lud Z'$ is $\NP$-hard by a reduction from 3-SAT. We report
some experiments and conclude in
Section~\ref{sec:exper-concl}. Missing proofs can be found in 
the Appendix.

\section{Preliminaries}
\label{sec:preliminaries}


Let $\Nat$ denote the set of natural numbers, $\Int$ the set of
integers and $\Rpos$ the set of non-negative reals. We denote the
power set of a set $S$ by
$\Pp(S)$. 
A \emph{clock} is a variable that ranges over $\Rpos$. Fix a finite
set of clocks $X$. A \emph{valuation} $v$ is a function which maps
each clock $x \in X$ to a value in $\Rpos$. Let $\Phi(X)$ denote the
set of \emph{clock constraints} $\phi$ formed using the following
grammar: $\phi := x \sim c~ \mid ~ x - y \sim c ~\mid~ \phi \land
\phi$, where $x, y \in X$, $c \in \Nat$ and ${\sim} \in \{<, \le, = ,
\ge, >\}$
Constraints of the form $x - y \sim c$ are called \emph{diagonal
  constraints}.
For a clock constraint $\phi$, we write $v \models \phi$ if the
constraint given by $\phi$ is satisfied by replacing each clock $x$ in
$\phi$ with $v(x)$.  For $\d\in\Rpos$, we write $v+\d$ for the
valuation defined by $(v+\d)(x)=v(x)+\d$ for all clocks $x$.  For a
set $R$ of clocks, we write $[R]v$ for the valuation obtained by
setting each clock $x \in R$ to $0$ and each $x \notin R$ to $v(x)$.

\begin{definition}[Timed Automata] A \emph{timed automaton} $\Aa$ is a
  tuple $(Q, X, \D, q_0, F)$ where $Q$ is a finite set of states, $X$
  is a finite set of clocks, $q_0 \in Q$ is the initial state, $F
  \incl Q$ is a set of accepting states and $\Delta \incl Q \times
  \Phi(X) \times \Pp(X) \times Q$ is the transition relation. Each
  transition in $\Delta$ is of the form $(q, g, R, q')$ where $g \in
  \Phi(X)$ is called the \emph{guard} of the transition and $R \incl
  X$ is the set of clocks that are said to be \emph{reset} at the
  transition.
\end{definition}

Timed automata with no diagonal constraints are called
\emph{diagonal-free}.
The semantics of timed automata is described as a transition system
over the space of its \emph{configurations}. A configuration is a pair
$(q, v)$ where $q \in Q$ is a state and $v$ is a valuation. There are
two kinds of transitions. \emph{Delay} transitions are given by $(q,
v) \to^{\d} (q, v+\d)$ for each $\d \in \Rpos$, and \emph{action}
transitions are given by $(q, v) \to^{t} (q', v')$ for each transition
$t \in \Delta$ of the form $(q, g, R, q')$, if $v \models g$ and $v' =
[R]v$.
The initial configuration is $(q_0, \vali)$ where $\vali$ denotes the
valuation mapping each clock to $0$.  Note that the above transition
system is infinite. A \emph{run} of a timed automaton is an
alternating sequence of delay and action transitions starting from the
initial configuration: $(q_0, \vali) \to^{\d_0}\to^{t_0} (q_1, v_1)
\to^{\d_1}\to^{t_1} \cdots (q_n, v_n)$. A run of the above form is
said to be accepting if the last state $q_n \in F$.  The
\emph{reachability problem} for timed automata is the following: given
an automaton $\Aa$, decide if there exists an accepting run. This
problem is known to be $\PSPACE$-complete~\cite{Alur:TCS:1994}. As the
space of configurations is infinite, the main challenge in solving
this problem involves computing a finite (and as small as possible)
abstraction of the timed automaton semantics. 
In this section, we recall the reachability algorithm for this
diagonal free case. For the rest of the section we fix a timed
automaton $\Aa$.

Instead of working with configurations, standard solutions in timed
automata analysis work with sets of valuations. The ``successor''
operation is naturally extended to the case of sets. For every
transition $t$ of $\Aa$ and every set of valuations $W$, we have a
transition $\tto^t$ defined as follows: $(q, W) \tto^t (q',W') $ where
$W' = \{ v' ~\mid~ \exists v \in W, ~\exists \d \in \Rpos:~ (q,v)
\to^t \to^\d (q',v') \}$.  Note that in the definition we have a
$\to^\d$ following the $\to^t$. This ensures that the ${\tto}$
successors (where ${\tto} =\bigcup_{t\in\Delta}\tto^{t}$) are closed
under time successors. Moreover, the sets which occur during timed
automata analysis using the $\tto$ relation have a special structure,
and are called \emph{zones}. A zone is a set of valuations which can
be described using a conjunction of constraints of the form: $x \sim
c$ or $x - y \sim c$ where $x, y \in X$ and $c \in \Nat$. Zones can be
efficiently represented using Difference Bound Matrices (DBMs). To
each automaton $\Aa$, we associate a transition system consisting of
(state, zone) pairs: the \emph{zone graph} $\ZG(\Aa)$ is a transition
system whose nodes are of the form $(q,Z)$ where $q$ is a state of
$\Aa$ and $Z$ is a zone. The initial node is $(q_0, Z_0)$ with $Z_0 =
\{ \vali + \d~\mid~ \d \ge 0\}$. Transitions are given
by 
$\tto$.

\begin{lemma}
  \label{lem:zg-correct}
  The zone graph $\ZG(\Aa)$ is sound and complete for
  reachability~\cite{Daws:TACAS:1998}.
\end{lemma}

Although the zone graph is a more succinct representation than the
space of configurations, it could still be infinite. The reachability
algorithm employs \emph{simulation relations} between zones to obtain
a finite zone graph that is sound and complete\footnote{Existing
  reachability algorithms make use of what are known as abstraction
  operators~\cite{Behrmann:STTT:2006, Herbreteau:IandC:2016},
  which are based on simulation relations. Instead of abstractions, we choose to present the algorithm directly using
  simulations between zones.  }.





We start by defining this notion of simulations at the level of
configurations. A \emph{(time-abstract) simulation} between pairs of
configurations of $\Aa$ is a reflexive and transitive relation $(q, v)
\simmA (q',v')$ such that: $q = q'$; for every $(q, v) \to^\d (q,
v+\d)$ there exists $\d'$ such that $(q, v') \to^{\d'} (q, v'+\d')$
satisfying $(q,v + \d) \simmA (q,v'+\d')$; and if $(q,v) \to^t (q_1,
v_1)$, then there exists $(q,v') \to^t (q_1, v_1')$ satisfying $(q_1,
v_1) \simmA (q_1,v_1')$ for the same transition $t$. We say that
$(q,v)$ is simulated by $(q', v')$. We write $v \simmA v'$ if $(q, v)
\simmA (q, v')$ for all states $q$. Simulations can be extended to
relate zones in the natural way: we write $Z \simmA Z'$ if for all $v
\in Z$ there exists $v' \in Z'$ such that $v \simmA v'$.  A simulation
relation $\fleq$ is said to be finite 
if there exists $N \in \Nat$ such that for all $n > N$ and every
sequence of zones $\{ Z_1, Z_2, \dots, Z_n \}$,
there exists $i < j \le n$ such that $Z_j \fleq Z_i$.

\smallskip

\label{algo:reachability}
\textbf{Reachability algorithm.} The input to the algorithm is a timed
automaton $\Aa$. The algorithm maintains two lists Passed and Waiting,
and makes use of a finite simulation relation $\fleq$ between zones.
The initial node $(q_0, Z_0)$ is added to the Waiting list. The
algorithm repeatedly performs the following tasks:
\begin{description}
\item[Step 1.] If Waiting is empty, then return ``$\Aa$ has no
  accepting run''; else pick a node $(q, Z)$ from Waiting.
\item[Step 2.] For each successor $(q, Z) \tto (q_1, Z_1)$ such that
  $Z_1 \neq \emptyset$ perform the following operations: if $q_1$ is
  accepting, return ``$\Aa$ has an accepting run''; else check if
  there exists a node $(q_1, Z_1')$ in Passed or Waiting such that
  $Z_1 \fleq Z_1'$: if yes, ignore the node $(q_1, Z_1)$, otherwise
  add $(q_1, Z_1)$ to Waiting.
\item[Step 3.] Add $(q, Z)$ to Passed and proceed to Step 1.
\end{description}

\begin{theorem}
  \label{thm:reachability-algo-terminates}
  The reachability algorithm terminates with a correct answer.
\end{theorem}
\begin{proof}
  Termination follows from the fact that the algorithm uses a finite
  simulation relation. We now focus on correctness. When the algorithm
  returns ``$\Aa$ has an accepting run'', it has detected a path in
  $\ZG(\Aa)$ leading to an accepting state. By soundness of the
  zone graph (Lemma \ref{lem:zg-correct}), the answer is correct. When
  the algorithm returns ``$\Aa$ has no accepting run'', we need to
  ensure that it has not missed any paths in the zone graph due to the
  pruning arising out of $\fleq$. We will now show that for
  every node $(q, Z)$ in $\ZG(\Aa)$ there is a node $(q, Z')$ in the
  Passed list such that $Z \fleq Z'$.

  We prove this by induction on the length of the path to $(q, Z)$
  starting from the initial node. The initial node $(q_0, Z_0)$ is
  added to the Waiting list as the initialization step. Step 1 and 2
  would be done for this node, and since we are in the case where the
  algorithm terminates due to Step 1, we infer that Step 3 was
  performed for $(q_0, Z_0)$. This shows that $(q_0, Z_0)$ is in the
  Passed list, thereby proving the base case. Suppose the hypothesis
  is true for some node $(q, Z)$ of $\ZG(\Aa)$. Consider a
  successor $(q, Z) \tto^t (q_1, Z_1)$ in $\ZG(\Aa)$. By induction
  hypothesis, there is a node $(q, Z')$ in Passed with $Z \fleq Z'$. 
  Hence Step 2 was performed on $(q, Z')$ and a successor
  $(q, Z') \tto^t (q_1, Z_1')$ was computed. From the definition of
  simulations, we get $Z_1 \fleq Z_1'$. If $(q_1, Z_1')$ was 
  added to Passed, we are done. Otherwise, we know that there exists
  a node $(q_1, Z_1'')$ in Passed or Waiting such that $Z_1' \fleq
  Z_1''$. By transitivity of $\fleq$, we get $Z_1 \fleq
  Z_1''$. If $(q_1, Z_1'')$ is in Passed, we are done. Else, it
  was in the Waiting list. Since the algorithm terminates due to Step
  1 where the Waiting list is empty, we can infer that $(q_1, Z_1'')$
  was removed from Waiting and added to Passed in its
  corresponding Step 3. 
\end{proof}

The reachability algorithm relies on an operation $Z_1 \fleq Z_1'$,
where $\fleq$ is some finite simulation relation as defined
earlier.  It has been shown that for the simulation relation $\lu$
 of \cite{Behrmann:STTT:2006} which works for diagonal free
  automata, checking $Z \lu Z'$ can be done in time
  $O(|X|^2)$~\cite{Herbreteau:IandC:2016}. Hence in diagonal free
  timed automata, this simulation test is as efficient as checking
  normal inclusion
  $Z \incl Z'$. The successor computation can also be implemented in
$O(|X|^2)$~\cite{Zhao:IPL:2005} using DBMs. These matrices can also be
viewed as graphs. We recall this graph-based representation of zones
and some of its properties.

\begin{definition}[Distance graph]
  A \emph{distance graph} $G$ has clocks as vertices, with an
  additional special vertex $x_0$ representing constant $0$. Between
  every two vertices there is an edge with a \emph{weight} of the form
  $({\lleq}, c)$ where $c\in \mathbb{Z}$ and ${\lleq} \in \{\le, < \}$
  or $({\lleq},c) = (<,
  \infty)$.
  An edge $x\act{{}\lleq c} y$ represents a constraint
  $y-x\mathrel{\lleq} c$: or in words, the distance from $x$ to $y$ is
  bounded by $c$. We let $\sem{G}$ be the set of valuations of clock
  variables satisfying all the constraints given by the edges of $G$
  with the restriction that the value of $x_0$ is $0$.
\end{definition}

We will sometimes write $0$ instead of $x_0$ for clarity. An
arithmetic over the weights $(\lleq, c)$ can be defined as
follows~\cite{Bengtsson:Springer:2004}.
\begin{description}
\item \emph{Equality} $(\lleq_1,c_1) = (\lleq_2,c_2)$ if $c_1 = c_2$
  and ${\lleq}_1 = {\lleq}_2$.
\item \emph{Addition} $(\lleq_1,c_1) + (\lleq_2,c_2) = (\lleq,c_1 +
  c_2)$ where ${\lleq} = {<}$ iff either $\lleq_1$ or $\lleq_2$ is $<$.
\item \emph{Total order} $(\lleq_1,c_1) < (\lleq_2,c_2)$ if either
  $c_1 < c_2$ or ($c_1 = c_2$ and ${\lleq}_1 = {<} $ and ${\lleq}_2 =
  {\le}$).
\end{description}
This arithmetic lets us talk about the weight of a path as the sum of
the weights of its edges.

A cycle in a distance graph $G$ is said to be \emph{negative} if the
sum of the weights of its edges is at most $(<,0)$.
A distance graph is in \emph{canonical form} if there are no negative
cycles and the weight of the edge from $x$ to $y$ is the lower bound
of the weights of paths from $x$ to $y$.  Given a distance graph, its
canonical form can be computed by using an all-pairs shortest paths
algorithm like Floyd-Warshall's~\cite{Bengtsson:Springer:2004} in time
$\Oo(|X|^3)$ where $|X|$ is the number of clocks. Note that the number
of vertices in the distance graph is $|X| + 1 $.
A folklore result is that: a distance graph $G$ has no negative cycles
iff $\sem{G}\not=\es$.  Given two distance graphs $G_1, G_2$ (not
necessarily in their canonical form), we define $\min(G_1, G_2)$ to be
the distance graph obtained by setting for each $x \to y$ the minimum
of the corresponding weights in $G_1$ and $G_2$. For two distance
graphs $G_1$ and $G_2$, we have $\sem{\min(G_1, G_2)} = \sem{G_1} \cap
\sem{G_2}$.
\label{distance-graphs}

A simulation relation for timed automata with diagonal constraints was
proposed in \cite{Bouyer:2004:forwardanalysis}, but it has not been
used in the reachability algorithm since no algorithm for
the zone simulation test was known. 



\section{A new simulation relation in the presence of diagonal
  constraints}
\label{sec:new-simul-relat}

In this section, we introduce a new simulation relation $\lud$ which
extends the $\lu$ simulation of \cite{Behrmann:STTT:2006}.  For this,
we first assume that all guards in timed automata are rewritten in the
form $ x - y \lleq c$ or $x \lleq c$, where $c \in \mathbb{Z}$ and
${\lleq} \in \{<, \le\}$. We will also assume that $X$ is a set of
clocks including the $0$ clock.

\begin{definition}[LU-bounds]
  An \emph{$LU$ bounds function} is a pair of functions $L: X \times X
  \mapsto \mathbb{Z} \cup \{ \infty \}$ and $U: X \times X \mapsto
  \mathbb{Z} \cup \{-\infty\}$ mapping each clock difference $x - y$
  to a constant or $\infty$ or $-\infty$ such that the conditions
  below are satisfied (we write $L(x-y), U(x-y)$ for $L(x,y)$ and
  $U(x,y)$ respectively):
  \begin{itemize}
  \item either $L(x-y) = \infty$ and $U(x-y) = -\infty$, or $L(x - y)
    \le U(x - y)$ for all distinct pairs of clocks $x, y \in X$, 
  \item $L(x - 0) = 0$ and $U(0 - x ) = 0$ for all non zero clocks $x
    \in X$
  \end{itemize}
\end{definition}

The $L$ stands for \emph{lower} and $U$ stands for
\emph{upper}. Intuitively, each $LU$-bounds function corresponds to a
set of guards given by $ x - y \lleq c $ with $L(x - y) \le c \le U(x-
y)$. We will now define a simulation relation $\lud$ between
valuations parameterized by $LU$-bounds. The idea is to give a
relation $v \lud v'$ such that $v'$ satisfies all guards compatible
with the parameter $LU$ that $v$ satisfies. To achieve this, the
situation as illustrated in
Figure~\ref{fig:new-lu-d-illustration}
needs to be avoided. This is formalized by the following definition
and the subsequent lemma.

\begin{figure}[t]
\centering

%
%
%
%
%
%
\begin{tikzpicture}
\shade [inner color = gray!30, outer color=gray!60] (2.5, -0.1)
rectangle (0, 0.1);
\draw [thick] (0,0) -- (5,0);


\draw (2.5, -0.2) -- (2.5, 0.2);
\node at (2.5, -0.4) {\tiny $d$};

\fill (1.3, 0) circle (2pt);
\node at (1.3, 0.3) {\tiny $v(x)-v(y)$};

\fill (3.5, 0) circle (2pt);
\node at (3.5, 0.3) {\tiny $v'(x)-v'(y)$};

\node at (2.5, -0.8) {\tiny Do not relate $v$ and $v'$ if there is a
  guard $x - y \le d$ or $x - y < d$};
\end{tikzpicture}

\medskip

%
%
%
%
%
%
\caption{Black dots illustrate the values of $v(x) - v(y)$ and $v'(x)
  - v'(y)$. The value $v(x) - v(y)$
  satisfies the guard $x - y\lleq d$ but $v'(x) - v'(y)$ does not
  satisfy the same guard.}
\label{fig:new-lu-d-illustration}
\end{figure}

\begin{definition}[LU-preorder $\lud$]\label{def:lu-d-preorder} Let
  $LU$ be a bounds function. A valuation $v'$ simulates a valuation
  $v$ with respect to $LU$, written as $v \lud v'$, if for every pair
  of distinct clocks $x, y \in X$ the following hold:
  \begin{itemize}
  \item $v'(x) - v'(y) < L(x-y)$ if $v(x) - v(y) < L(x - y)$
  \item $v'(x) - v'(y) \le v(x) - v(y)$ if $L(x-y) \le v(x) - v(y) \le
    U(x - y)$
  \end{itemize}
  For a valuation $v$, we write $\lureg{v}$ for the set of all $v'$
  such that $v \lud v'$.
\end{definition}

\begin{lemma}
  \label{lem:lu-satisfies-same-guards}
  Let $x, y$ be distinct clocks in $X$, and $x - y \lleq c$ with $c
  \in \mathbb{Z}$ be a guard. Let $LU$ be a bounds function such that
  $L(x-y) \le c \le U(x-y)$.  Then, for every pair of valuations $v,
  v'$ such that $v \lud v'$, if valuation $v \models x - y \lleq c$
  then $v' \models x - y \lleq c$.
\end{lemma}
\begin{proof}
  Assume $v \lud v'$ and $ v \models x - y \lleq c$.  If $v(x)
  - v(y) < L(x-y)$ then $v'(x) - v'(y) < L(x-y)$. Hence, $v'(x) -
  v'(y) < c$ as $L(x - y) \le c$.  If $L(x-y) \le v(x) - v(y) \le U (x
  - y)$ then $v'(x) - v'(y) \le v(x) - v(y)$ and hence $v'(x) - v'(y)
  \lleq c$.  If $v(x) - v(y) > U (x-y)$ then as $U(x-y) \ge c$, we get
  $v \not\models x - y \lleq c$, contradicting our assumption on $v$.
 %
\end{proof}

The next lemmas show that time delay preserves $\lud$ from two
valuations $v$ and $v'$ with $v \lud v'$. In fact, it is strong in the
sense that if we delay $\d$ from $v$, then the same delay from $v'$
satisfies the $LU$ preorder conditions.

\begin{lemma}
  \label{lem:lud-preserves-delay}
  Let $LU$ be a bounds function. For every pair of valuations $v$ and
  $v'$, if $v \lud v'$, then $v + \d \lud v'+ \d$ for all $\d \ge 0$.
\end{lemma}
\begin{proof}
  We need to show that the conditions of
  Definition~\ref{def:lu-d-preorder} hold for $(v+\d)(x) - (v+\d)(y)$ and
  $(v'+\d)(x) - (v'+\d)(y)$ using the fact that $v \lud v'$.

  When $x, y$ are distinct clocks with neither of them being $0$, we
  have $(v + \d)(x) - (v+\d)(y) = v(x) - v(y)$ and $(v'+\d)(x)-(v'+\d)(y) =
  v'(x)-v'(y)$. Conditions given in Definition \ref{def:lu-d-preorder}
  are automatically satisfied for such $x,y$ since $v \lud v'$.

  When $y$ is the $0$ clock, the first case in Definition
  \ref{def:lu-d-preorder} cannot arise as $L(x - 0) = 0$. Suppose $L(x
  - 0) \le (v+\d)(x) \le U(x-0)$. Then $L(x - 0) \le v(x) \le U(x - 0)$
  and hence $v'(x) \le v(x)$ as $v \lud v'$. This gives $(v' + \d)(x)
  \le (v+\d)(x)$.

  When $x$ is the $0$ clock, note that $-(v+\d)(y) = -v(y) - \d$,
  $-(v'+\d)(y) = -v'(y) - \d$.  Additionally, as $U(0-y) = 0$, either
  both $-v(y) < L(0-y)$ and $-v'(y)<L(0-y)$ or $-v'(y) \le -v(y)$.
  This entails that the conditions for $\lud$ hold for $-(v+\d)(y)$ and
  $-(v'+\d)(y)$.
\end{proof}

The next lemma shows that resets preserve $\lud$ under certain
conditions on $LU$.

\begin{lemma}
  \label{lem:lu-satisfies-resets}
  Let $LU$ be a bounds function satisfying $U(x-0)\ge U(x-y)$ for all
  $y \in X$ and $L(0-y) \le L(x-y)$ for all $x \in X$. Then, if $v
  \lud v'$, then $[R]v \lud [R]v'$ for every $R \subseteq X$.
\end{lemma}
\begin{proof}
  We need to show that the conditions of
  Definition~\ref{def:lu-d-preorder} hold for $a=([R]v)(x) - ([R]v)(y)$ and
  $a'=([R]v')(x) - ([R]v')(y)$ using the fact that $v \lud v'$.

  When $x \notin R$ and $y \notin R$, the conditions hold automatically since $a
  = v(x) - v(y)$ and $a' = v'(x) - v'(y)$.

  When $x \notin R$ and $y \in R$, then $a=v(x)-v(0)$, and $a' = v'(x)-v'(0)$.
  Since $L(x-0) = 0$, either $v'(x) \le v(x) \le U(x - 0)$ or $v(x) > U(x - 0)$.
  Since $U(x - 0) \ge U(x - y)$, we get $v'(x) \le v(x)$ if $v(x) \le U(x - y)$.
  This proves the conditions.

 When $x \in R$ and $y \notin R$,  then $a=v(0)-v(y)$, and $a' = v'(0)-v'(y)$.
 Then, either $a<L(0-y)$ and $a'<L(0-y)$ or $L(0-y)\leq a\leq U(0-y)=0$ and
 $a'\leq a$ (notice that $a>U(0-y)=0$ is not possible). Therefore, either
 $a'<L(x-y)$ or $a'\leq a$ and the conditions hold.
\end{proof}


The $LU$ preorder can be extended to configurations: $(q, v) \lud (q,
v')$ if $v \lud v'$. The above three lemmas give the necessary
ingredients to generate an $LU$ bounds function from a timed automaton
$\Aa$ such that the associated $LU$ preorder is a simulation on its
space of configurations.

Let $\Gg$ be a set of constraints.  We construct a new set $\Ggbar$
from $\Gg$ in the following way:
\begin{itemize}
\item Add all the constraints of $\Gg$ to $\Ggbar$
\item For each clock $x \in X$, add the constraints $x \le 0$ and $-x
  \le 0$ to $\Ggbar$
\item For each constraint $x - y \lleq c \in \Gg$, add the constraints
  $x \lleq c$ and $-y \lleq c$ to $\Ggbar$
\item Remove all constraints of the form $x \lleq c_1$ where $c_1 \in
  \mathbb{R}_{<0}$ and constraints of the form $ -x \lleq c_2$ where
  $c_2 \in \mathbb{R}_{>0}$ from $\Ggbar$.
\end{itemize}

We define an $LU$-bounds function on $\Ggbar$ in the natural way: for
each pair of clocks $x,y \in X$, we set $L (x-y) = \min\{c \mid x - y
\lleq c \in \Ggbar \}$ and $U (x-y) = \max\{c \mid x - y \lleq c \in
\Ggbar \}$.  If there are no guards of the form $x - y \lleq c$ in
$\Ggbar$, then we set $L (x - y)$ to be $\infty$ and $U(x - y)$ to be
$-\infty$.  Note that since $\Ggbar$ contains the constraints $x \le
0$ and has no constraints $x \lleq c$ where $c \in \mathbb{R}_{<0}$,
$L (x-0) = 0$ for all $x \in X$.  Similarly, $U(0 - x) = 0$ for all $x
\in X$.  For a timed automaton $\mathcal{A}$, let $\Ga$ be the set of
guards present in $\Aa$.  The \emph{$LU$-bounds of $\mathcal{A}$} is
the $LU$-bounds function defined on $\overline{\Ga}$. The next theorem
follows from Lemmas~\ref{lem:lu-satisfies-same-guards},
\ref{lem:lud-preserves-delay} and \ref{lem:lu-satisfies-resets}.

\begin{theorem}
  For every timed automaton $\Aa$, the relation $\lud$ obtained from
  the $LU$-bounds of $\Aa$ is a simulation relation on its
  configurations.
\end{theorem}

We use this simulation relation extended to zones in the reachability
algorithm, as described in Page \pageref{algo:reachability}. To do so,
we need to give an algorithm for the simulation test $Z \lud Z'$, and
show that $\lud$ is finite. Correctness and termination follow from
Theorem~\ref{thm:reachability-algo-terminates}. We first describe the
simulation test, and then prove finiteness.
Observe that $Z \not \lud Z'$ iff there exists $v \in Z$ such that
$\lureg{v} \cap Z' = \emptyset$. We give a distance graph
representation for $\lureg{v}$.
\begin{definition}[Distance graph for $\lureg{v}$]
  \label{def:gvlu-distance-graph}
  Given a valuation $v$ and an $LU$ bounds function, we construct
  distance graph $\gvlu$ as follows. For every pair of distinct clocks
  $x,y \in X$, add the edges:
  \begin{itemize}
  \item $y \xra{} x$ with weight $(<, L(x - y))$, if $v(x) - v(y) <
    L(x-y)$,
  \item $y \xra{} x$ with weight $(\le, v(x) - v(y))$, if $L(x-y) \le
    v(x) - v(y) \le U(x-y)$.
  \end{itemize}
\end{definition}

Using Definition \ref{def:lu-d-preorder} we can show that
$\sem{\gvlu}$ equals $\lureg{v}$. The properties of distance graphs as
described in Page~\pageref{distance-graphs} then lead to the following
theorem. 

\begin{theorem}
  \label{thm:alu-inclusion}
  Let $Z, Z'$ be zones such that $Z'$ is non-empty, and let $LU$ be a
  bounds function. Let $G_{Z'}$ be the canonical distance graph of
  $Z'$. Then, $Z \not \lud Z'$ iff there is a valuation $v \in Z$ and
  a negative cycle in $\min(\gvlu, G_{Z'})$ in which no two
  consecutive edges are from $G_{Z'}$.
\end{theorem}

A witness to the fact that $Z \not \lud Z'$ is therefore a $v \in Z$
and a negative cycle of a certain shape given by
Theorem~\ref{thm:alu-inclusion}. Existence of such a witness can be
encoded as satisfiability of a formula in linear arithmetic.  This
gives an NP procedure. A satisfying assignment to the formula reveals
a valuation $v \in Z$ and a corresponding negative cycle across
$\gvlu$ and $G_{Z'}$. Although there is no fixed bound on the length
of this negative cycle (contrary to the diagonal free case), note that
each $y \to x$ edge from $\gvlu$ in the negative cycle needs to have a
finite $U(x-y)$ or $L(x-y)$ constant \textbf{(apart from $x \to 0$
  edges)}. If for an automaton, many of the edges have $\infty$ or
$-\infty$ as their $L$ or $U$ respectively (which we believe occurs
often in practice, as there could be no relevant diagonal constraint
over this edge) then this simulation test would need to enumerate only
a small number of negative cycles.

The final step is to show that $\lud$ is finite. We make use of a
notation: we write $\down{Z}$ to be the set of valuations $u$ such
that $u \lud v$ for some $v \in Z$. Note that $Z \not \lud Z'$ implies
$\down{Z} \neq \down{Z'}$.

\begin{theorem}
  \label{thm:lud-has-finite-range}
  The simulation relation $\lud$ is finite for every $LU$ bounds
  function.
\end{theorem}
\begin{proof}
  We will first show that for any zone $Z$, $\down{Z}$ is a union of
  \emph{$d$-regions} (parameterized by $LU$) which are defined below.
  We will subsequently show that there are only finitely many
  $d$-regions. The observation that $Z \not \lud Z'$ implies $\down{Z}
  \neq \down{Z'}$ then proves the theorem.

  Given a valuation $v$ and $LU$-bounds function, we define the
  following relations over pairs of clocks:
  \begin{itemize}
  \item $y \tightone x$ if $v(x) - v(y) < L(x-y)$
  \item $y \tighttwo x$ if $L(x-y) \le v(x) - v(y) \le U(x-y)$
  \end{itemize}

  A $d$-region $R$ is a set of valuations that satisfies the
  following:
  \begin{itemize}
  \item all valuations in $R$ have the same $\tightone$ and
    $\tighttwo$ relations.
  \item for every subset $S = \{ y_1 \tighttwo x_1, y_2 \tighttwo x_2,
    \dots, y_k \tighttwo x_k \}$ of ordered pairs of clocks, every
    valuation in $R$ satisfies one of the following constraints:
    either $\left (\sum\limits_{i = 1}^{i = k} x_i - y_i = c \right)$
    or $c-1 < \left (\sum\limits_{i = 1}^{i = k} x_i - y_i \right) <
    c$ for an integer $c$ satisfying $\sum\limits_{i = 1}^{i = k}
    L(x_i - y_i) \le c \le \sum\limits_{i = 1}^{i = k} U(x_i - y_i)$.
  \end{itemize}

  We will now show that if a $d$-region $R$ intersects $\down{Z}$ then
  $R \incl \down{Z}$. Let $v \in R$ be such that $v \in \down{Z}$. Let
  $v'$ be another valuation in $R$. Suppose $v' \notin \down{Z}$.  Then
  $\lureg{v'} \cap Z = \emptyset$.  That is, $\min(\gvplu, G_{Z})$ has
  a negative cycle; let us call it $N_{v'}$.  Let $N_{v}$ be the cycle
  $N_{v'}$ with the edges coming from $\gvplu$ replaced with the same
  edges from $\gvlu$.  We want to show that $N_{v}$ is negative.
  Since, $v$ and $v'$ come from the same region $R$, we have:
  \begin{itemize}
    \item The weight of a type 1 edge $y_i \tightone x_i$ is $(<,L(x_i - y_i))$
    in both $N_v$ and $N_{v'}$.  Let $(<,S_1)$ be the sum of the weights of the
    type 1 edges.  This sum is the same in $N_v$ and $N_{v'}$.

    \item We let $(\leq,S_2)=(\leq,\sum_{i} v(x_i) - v(y_i))$ and
    $(\leq,S'_2)=(\leq,\sum_{i} v'(x_i) - v'(y_i))$ be the sum of the weights of
    type 2 edges $y_i \tighttwo x_i$ in $N_v$ and $N_{v'}$ respectively.  Then,
    for some integer $c$, either $S_2=S'_2=c$ or $c-1 < S_2 < c$ and $c-1 < S'_2
    < c$.
  \end{itemize}
  Also the edges coming from $G_{Z}$ have the same weight in $N_v$ and $N_{v'}$.
  Call $(\lleq_3,S_3)$ the sum of the weights of the edges coming from $G_Z$.
  Finally, let $(\lleq,S=S_1+S_2+S_3)$ and $(\lleq,S'=S_1+S'_2+S_3)$ be the
  weights of $N_v$ and $N_{v'}$ respectively.  Since $N_{v'}$ is negative,
  $(\lleq,S')$ is at most $(<,0)$.  Now, $S_1$ and $S_3$ are integers, and using
  the relation between $S_2$ and $S'_2$, we deduce that $N_v$ is also negative.
  This entails $\lureg{v} \cap Z = \emptyset$, and
  contradicts the assumption that $v \in \down{Z}$.  Hence we get $R \subseteq \down{Z}$, thereby showing that each $\down{Z}$ is a union
  of $d$-regions.

  Each $d$-region depends only on the orientation of the $\tightone$
  and $\tighttwo$ relations and the value of $c$.  Since number of
  clocks is finite, the number of possible orientations of $\tightone$
  and $\tighttwo$ is finite.  For each such orientation, the possible
  values for $c$ is finite.  Thus there are only finitely many
  $d$-regions.
\end{proof}

\section{Algorithm for $Z \not \lud Z'$}
\label{sec:algor-zone-incl-LU}



Theorem~\ref{thm:alu-inclusion} gives a witness for the fact that
$Z \not \lud Z'$. In this section, we encode the existence of
this witness as an SMT formula over linear arithmetic. For clarity of
exposition, we will also restrict to timed automata having no strict
constraints as guards, that is, every guard is of the form $x - y \le
c$ or $x\leq c$. This would in
particular imply that in the zones obtained during the forward
analysis, there will be no strict constraints.

\begin{definition}[Satisfiability modulo Linear Arithmetic]
  Let $\Prop$ be a set of propositional variables, and $\Vars$ a set
  of variables ranging over reals. An atomic term is a constraint of
  the form $c_1 x_1 + c_2 x_1 + \dots + c_k x_k \sim d$ where $c_1,
  \dots, c_n, d \in \Int$ and $x_1, x_2,  \dots, x_k \in \Vars$ and
  ${\sim} \in \{ \le, <, =, >, \ge\}$. A formula in \emph{linear
    arithmetic} is a boolean combination of propositional variables
  and atomic terms. Formula $\phi$ is
  satisfiable if there exists an assignment of boolean values to
  propositions, and real values to variables in $\Vars$ s.t.
  replacing every occurence of the variables and propositions by the
  assigment evaluates $\phi$ to true.
\end{definition}

\begin{lemma}
\label{lem:linear-arithmetic-in-np}
  Satisfiability of a formula in linear arithmetic is in $\NP$.
\end{lemma}
\begin{proof}
  Given a formula in linear arithmetic, a certificate would be an
  assignment to all atomic terms and propositional variables. The
  conjunction of all atomic terms which are true would form a system
  of linear inequalities. Deciding if this system is consistent can be
  done in polynomial time (can be seen as a linear program with a
  dummy objective function, and linear programming can be solved in
  polynomial time~\cite{Karmarkar:1984:NPA}).
\end{proof}

Fix two zones $Z, Z'$ and a bounds function $LU$. Zones $Z$
and $Z'$ are given by their canonical distance graphs $G_Z$ and
$G_{Z'}$. We write $c_{yx}$ for the weight of the edge $y \to x$ in
$G_Z$ and $c'_{yx}$ for the weight of $y \to x$ in $G_{Z'}$. Further
we assume that the set of clocks is $\{ \0, \1, \dots, \n \}$.
The final formula will be obtained by constructing suitable
intermediate subformulas as explained below:
\begin{description}
\item[Step 1.] Guess a $v \in Z$.
\item[Step 2.] Guess a subset of edges $y \to x$ which forms a cycle
  (or a disjoint union of cycles).
\item[Step 3.] Guess a colour for each edge $y \to x$ in the cycle:
  red or blue. No two consecutive edges in the cycle can both
  be red.  Red edges correspond to edges from $G_{Z'}$.
  Blue edges correspond to edges from $\gvlu$.
\item[Step 4.] Assign weights to each edge $y \to x$: if it is
  coloured red, the weight is $c'_{yx}$ (edge weight of $G_{Z'}$).
  If the edge $y\to x$ is blue, assign weight according to the following cases:
  \begin{itemize}
    \item $w_{yx} = (<, L(x - y))$ if $v(x) - v(y) < L(x - y)$ 
    \item $w_{yx} = (\le, v(x) - v(y))$ if $L(x - y) \le v(x) - v(y) \le U(x
      - y)$
  \end{itemize}
  Add up the weights of all the edges (the comparison $<$ or $\le$
  component of the weight can be maintained using a boolean). If there
  are no strict edges (that is with weight $<$) in the chosen cycle,
  check if the sum is $ < 0$. Else, check if the sum is $\le  0$.
\end{description}

\textbf{Formula for Step 1.} We first guess a valuation $v \in
Z$. We use real variables $v_0, v_1, \dots, v_n$ to denote a
valuation. These variables should satisfy the constraints given by
$Z$:
\begin{align}
  \label{eq:v-in-Z}
  v_0 = 0 ~\text{ and } \Land_{x,y \in \{0, \dots, n\}} v_x - v_y \le
  c_{yx}
\end{align}
Call the above formula $\Phi_1(\bar{v})$ where
$\bar{v}=(v_1,\ldots,v_n)$.  A satisfying assignment to $\Phi_1$
corresponds to a valuation in $Z$.

\textbf{Formula for Step 2.}
We now need to guess a set of edges of the form $y \to x$ which forms
a cycle, or a disjoint union of simple cycles. We will also ensure
that no vertex appears in more than one cycle. We will use boolean
variables $e_{ij}$ for $i, j \in \{0, \dots, n\}$ and $i \neq j$.

The cycle must be non-empty.
\begin{align}
  \bigvee_{0\leq i,j\leq n, j\neq i} e_{ij}
  \label{eq:cycle-non-empty}
\end{align}
If we pick an incoming edge to a clock, then we need to pick an
outgoing edge.
\begin{align}
  \label{eq:cycle-incoming-outgoing}
  \Land_{0 \le i \le n} ~~\Big(~\Lor_{0 \le j \le n, j \neq i}
  e_{ji}~\Big) ~\implies~ \Big(~\Lor_{0 \le j \le n, j \neq i}
  e_{ij}~\Big)
\end{align}
We do not pick more than one outgoing or incoming edges for each
clock.
\begin{align}
  \label{eq:cycle-no-two-incoming-outgoing}
  \Land_{0 \le i \le n} ~~ \Land_{\begin{scriptsize}\begin{array}{c} 0
        \le j,k \le n \\ j \neq k \\ i \neq j, i \neq
        k \end{array}\end{scriptsize}} \neg (e_{ij} \land e_{ik})~
  \land ~ \neg (e_{ji} \land e_{ki})
\end{align}

Conjunction of (\ref{eq:cycle-non-empty},
\ref{eq:cycle-incoming-outgoing},
\ref{eq:cycle-no-two-incoming-outgoing}) gives a formula
$\Phi_2(\bar{e})$ over variables $\bar{e} = \{e_{12}, \dots, e_{n
  n-1}\}$.

\begin{lemma}\label{lem:step2}
  Let $\s_2: \bar{e} \mapsto \{\true, \false\}$ be an assignment which
  satisfies $\Phi_2$. Then the set of edges $\{x \to y\}$ such that
  $\s_2(e_{xy})$ is true forms a vertex-disjoint union of cycles.
\end{lemma}

\textbf{Formula for Step 3.}
To colour the edges of the cycle formed by $e_{ij}$, we will use
boolean variables $r_i$ for $0\leq i\leq n$ to color the source of the red
edges.
Once the red edges are determined, the blue edges are also uniquely
determined.
Only edges chosen by $\bar{e}$ are colored red, and no two consecutive edges can be coloured red.
\begin{align}
  \label{eq:LU-colour-red}
  \bigwedge_{0\leq i\leq n}
  \Big( r_i &\implies \bigvee_{0\leq j\leq n} e_{ij} \wedge \neg r_j \Big)
\end{align}
Then, red edges are edges with corresponding source $i$ satisfying $r_i$.
So for all $i,j\in\{0,\dots,n\}$ with $i\neq j$ we introduce the macro
$\red_{ij} := e_{ij} \wedge r_i$. 
Blue edges are those that have been chosen for the cycle and have
not been coloured red: $\blue_{ij} := e_{ij} \wedge \neg \red_{ij}$.
Each blue edge should satisfy one of the two conditions mentioned in
Definition~\ref{def:gvlu-distance-graph}.
\begin{align}
  \label{eq:LU-colour-blue}
  \Land_{i,j \in \{ 0, \dots, n\}, i \neq j} \blue_{ij} \implies  v_j - v_i \le U(j-i) 
\end{align}

Conjunction of (\ref{eq:LU-colour-red}) and
(\ref{eq:LU-colour-blue})
gives formula $\Phi_3$.

\begin{lemma}
  Let $\s_3$ be an assignment to variables $\bar{v}$, $\bar{e}$ and
  $\bar{r}$.
  Suppose $\s_3$ is a satisfying assignment for $\Phi_1 \land \Phi_2 \land
  \Phi_3$.  Then, the set of edges with $\s_3(e_{ij})$ being true forms a
  collection of vertex disjoint cycles formed from union of edges from
  $\gvlu$ and $G_{Z'}$ for some $v\in
  Z$.
\end{lemma}

\textbf{Formula for Step 4.} The last step is to add up weights of the red and blue edges. We
make use of real-valued variables $w_{i}$ for each source $i$ of an edge.
We associate weights of red and blue edges.
\begin{align}
  \label{eq:red-blue-weights}
  \Land_{i,j \in \{0, \dots, n\}, i\neq j} \hspace{-5mm}
  (\red_{ij} \implies w_{i} = c'_{ij}) \land &
  ((\blue_{ij} \wedge condition_1) \implies (w_i = L(j - i) \wedge \mathrm{~strict})) & \nonumber\\[-2ex] \land &
  ((\blue_{ij} \wedge condition_2) \implies w_i = v_j - v_i )
\end{align}
where, $condition_1 := v_j - v_i < L(j - i)$ and $condition_2 := L(j - i) \le v_j - v_i \le U(j - i)$.

Uncoloured edges take weight $0$,
\begin{align}
  \label{eq:no-colour-weight-0}
  \Land_{0\leq i\leq n}
  \Big( \bigwedge_{0\leq j\leq n} \neg e_{ij} \Big) \implies (w_{i} =
  0)
\end{align}
The final formula checks if the sum of the weights is at most $(<,0)$.
\begin{align}
  \label{eq:negative-cycle}
  \textstyle ((\sum_{0\leq i\leq n} ~w_{i}) < 0) \vee [ ~\mathrm{strict}~ \wedge ((\sum_{0\leq i\leq n} ~w_{i}) = 0) ]
\end{align}
Conjunction of (\ref{eq:red-blue-weights}),
(\ref{eq:no-colour-weight-0}) and (\ref{eq:negative-cycle}) gives
formula $\Phi_4$.
The final formula is $\Phi=\Phi_1 \land \Phi_2 \land \Phi_3 \land \Phi_4$.

\begin{theorem}
  Formula $\Phi$ as constructed above is satisfiable iff $Z \not \lud
  Z'$. 
\end{theorem}

Note that there are $\Oo(n)$ real variables $v_i$, $w_i$, and
$\Oo(n^2)$ booleans $e_{ij}, r_i$. Given the representations of $Z,
Z'$ and the $LU$ bounds, the entire formula $\Phi$ can be computed in
$\Oo(n^3)$, with formula (\ref{eq:cycle-no-two-incoming-outgoing}) 
taking the maximum time. This gives an $\NP$ procedure for $Z \not
\lud Z'$ (c.f. Lemma~\ref{lem:linear-arithmetic-in-np}).



\section{Checking $Z \not\lud Z'$ is $\NP$-hard}
\label{sec:hardness}

We will consider a special case of $\lud$, which already turns out to
be hard. Let $M \ge 0$ be a natural number. Consider the $LU$ bounds
functions obtained as $L(x - y) = -M$, $U(x - y) = M$ for all non-zero
pairs of clocks $x, y$; and $L(0 - x) = -M$, $U(x - 0) = M$ for all
non-zero clocks $x$. For notational convenience we denote by $\md$ the
simulation arising out of these $LU$ bounds.

\begin{lemma}
  $\md$ is a bisimulation for every $M \ge 0$.
\end{lemma}
\begin{proof}
  Note that if $-M \le v(x) - v(y) \le M$, then $-M \le v(y) - v(x)
  \le M$. Therefore, from Definition~\ref{def:lu-d-preorder} and the
  description of $\md$ given above, we can infer that $v \md v'$ if
  for all distinct $x,y$ (denoting $a = v(x) - v(y)$ and $a'=v'(x) -
  v'(y)$):
  \begin{itemize}
  \item either both $a'$ and $a$ are $< -M$
  \item or $-M \le a' = a \le M$
  \item or both $a'$ and $a$ are $> M$.
  \end{itemize}
  By symmetry we get $v \md v'$ iff $v' \md v$, showing that $\md$ is
  a bisimulation.
\end{proof}

Thanks to the above lemma, the $\md$ relation is an equivalence over
valuations. We will write $\regeq_M$ for $\md$, and $\reg{v}$ for the
set of valuations $v'$ such that $v \regeq_M v'$. The relation $\md$
can be extended to zones as in Page~\ref{algo:reachability}. With this
definition, we get that $Z \not \md Z'$ iff there exists $v \in Z$
such that for all $v' \regeq_M v$, we have $v' \not \in Z'$. The goal
is to show that deciding $Z \not \md Z'$ is $\NP$-hard. We
describe some notation and technical results before 
proceeding to the hardness proof.

We make use of a notion of \emph{tightness} between clocks which gets
induced by the $\regeq_M$ equivalence. Let $v$ be a valuation. Two
clocks $x_i$ and $x_j$ are said to be \emph{tight in $v$} if $-M \le
v(x_j) - v(x_i) \le M$. We denote this by $x_i \tight x_{j}$ (can be
read as $x_i$ and $x_{j}$ are tied to each other). Let $\tight^*$ (can
again be read as the tight relation) denote the reflexive and
transitive closure of $\tight$. Note that $\tight^*$ is an equivalence
over clocks. Moreover, when $v \regeq_M v'$, the equivalence classes
of $\tight^*$ in $v$ and $v'$ are identical. We say that a zone $Z$ is
\emph{topologically closed} if every edge $x \xra{} y$ in the
canonical distance graph of $Z$ has weight of the form $(\le, c)$ with
$c \in \Int$, or $(<, \infty)$.  A valuation $v$ mapping each $x$ to
an integer is said to be an \emph{integral valuation}. The next
proposition says that for certain topologically closed zones $Z$ and
$Z'$, if $Z \not \md Z'$ then there is an integral valuation as a
witness to this non-simulation. 

\begin{proposition}
  \label{prop:topo-closed-closure-inclusion}
  Let $Z$ be a topologically closed zone such that the $\tight^*$
  equivalence classes of every valuation in $Z$ are the same. Let $Z'$
  be a zone with $Z \not \md Z'$. Then, there exists an
  integral valuation $u \in Z$ such that $\reg{u} \cap Z'$ is empty.
\end{proposition}





\subsection{Reduction from 3-SAT}

Consider the decision problem which takes as inputs two zones $Z, Z'$
and outputs whether $Z \not \md Z'$.  We will give a
 polynomial time reduction from 3-SAT to this decision problem,
 showing that it is $\NP$-hard.

\textbf{Notation.} Let $\var$ be a finite set of propositional
variables. A \emph{literal} is either a variable $p$ or its negation
$\neg p$, and a 3-\emph{clause} is a disjunction of three literals
$(l_1 \lor l_2 \lor l_3)$. A 3-CNF formula is a conjunction of
3-clauses. For a literal $l$, we write $\var(l)$ for the variable
corresponding to $l$. For a 3-CNF formula $\phi$, we write
$\var(\phi)$ for the variables present in $\phi$. An
\textit{assignment} to a 3-CNF formula $\phi$ is a function from
$\var(\phi)$ to $\{\ttt,\fff\}$. 
For a clause $C$ and an assignment $\s$, we write $\s \models C$ if
substituting $\s(p)$ for each variable $p$ occurring in $C$ evaluates
the clause to true. For a formula $\phi$ and an assignment $\s$, we
write $\s \models \phi$ if all clauses of $\phi$ evaluate to true
under $\s$. A formula $\phi$ is said to be \emph{satisfiable} if there
exists an assignment such that $\s \models \phi$. For the rest of the
section, fix a 3-CNF formula $\varphi := C_1 \land C_2 \land \cdots
\land C_N$.  Let $\clauses(\varphi)$ be the set $\{C_i ~|~ i \in
\{1,\dots, N\}\}$. 

We start with the idea for the reduction. We know that $\varphi$ is
satisfiable iff there \emph{exists} an assignment $\sigma$ such that
\emph{for all} $C \in \clauses(\varphi): \sigma \models
C$. Correspondingly, we know that $Z \not \fleq_M Z'$ iff there
\emph{exists} a $v \in Z$ such that \emph{for all} $v' \regeq_M v: v'
\notin Z'$. Given $\varphi$, we want to construct two topologically
closed zones $Z, Z'$ such that $\varphi$ is satisfiable iff $Z
\not\fleq_M Z'$. We want the (potential) $v \in Z$ for which
every $v' \regeq_M v$ satisfies $v' \not\in Z'$ to encode the
(potential) satisfying assignment for $\varphi$. In essence:
valuations in $Z$ should encode assignments, the equivalent valuations
$v'$ should encode clauses and the fact that $v' \not \in Z'$ should
correspond to the chosen clause being true.  We now proceed with the
details of the construction. For each literal $l_i^j$ of $\varphi$, we
add three clocks $x_i^j, y_i^j, z_i^j$. There are $N+1$ additional
clocks $r_0, r_1, \dots, r_N$. We will assume an arbitrary constant $M
> 3$. Figure \ref{fig:hardness} illustrates the construction.


\textbf{Construction of $Z$.} Zone $Z$ is described by three sets of
constraints. The first set of constraints are between clocks of each
literal. For every $i \in \{1, \dots, N\}$ and $j \in \{1, 2,
3\}$: 
\begin{align}
  \label{eq:within-a-block}
  y_i^j - x_i^j \geq 1 \quad \text{ and } \quad z_i^j - y_i^j \geq 1
  \quad \text{ and } \quad z_i^j - x_i^j = 3
\end{align}
The second set of constraints relates the distance between clocks of
different literals. In addition, we use the $r_i$ clocks as separators
between clauses. For $i \in \{1, \dots, N\}$:
\begin{align}
  \label{eq:between-blocks}
  x_i^1 - r_{i-1} = 2M - 3 ~\text{ and }~ x_i^{j+1} - z_i^j = 2M - 3
  \text{ for } j \in \{1, 2\} ~\text{ and }~ r_{i} - z_i^3 = 2M
\end{align}

\begin{figure}[t]
\centering
\begin{tikzpicture}[vertex/.style={circle, fill, inner sep=1pt}]




\shade [rounded corners, inner color=gray!5, outer color=gray!20] (-0.2,
3) rectangle (14.2, 5.5);

\shade [rounded corners, inner color=gray!5, outer color=gray!20] (-0.2,
-0.5) rectangle (14.2, 2.6);

\node at (0.5, 5.3) {\scriptsize \textbf{Zone} $Z$};

\begin{scope}[yshift = 0cm]

\foreach \x in {0.5, 2.5, 4.5, 7.5, 9.5, 11.5}
\node at (\x, 4) {\tiny$2M - 3$};

\node at (6.5,4) {\tiny$2M$};
\node at (13.5,4) {\tiny$2M$};

\begin{scope}[thick]

\draw [->, bend left, color = red] (1.5,4) to (1.05, 4) {};
\draw [->, bend left, color = red] (2,4) to (1.55, 4) {};
\draw [->, bend left, color = red] (3.5,4) to (3.05, 4) {};
\draw [->, bend left, color = red] (4,4) to (3.55, 4) {};
\draw [->, bend left, color = red] (5.5,4) to (5.05, 4) {};
\draw [->, bend left, color = red] (6,4) to (5.55, 4) {};
\draw [->, bend left, color = red] (8.5,4) to (8.05, 4) {};
\draw [->, bend left, color = red] (9,4) to (8.55, 4) {};
\draw [->, bend left, color = red] (10.5,4) to (10.05, 4) {};
\draw [->, bend left, color = red] (11,4) to (10.55, 4) {};
\draw [->, bend left, color = red] (12.5,4) to (12.05, 4) {};
\draw [->, bend left, color = red] (13,4) to (12.55, 4) {};

\draw [rounded corners = 5] (5.5,4) to (5.5,4.5) to (8.5, 4.5) to (8.5,4) {};
\draw [rounded corners = 5] (1.5,4) to (1.5,5) to (12.5, 5) to (12.5,4) {};

\draw [rounded corners = 2, color = red] (2,4) to (2,3.5) to (1, 3.5) to (1,4) {};
\draw [rounded corners = 2, color = red] (4,4) to (4,3.5) to (3, 3.5) to (3,4) {};
\draw [rounded corners = 2, color = red] (6,4) to (6,3.5) to (5, 3.5) to (5,4) {};
\draw [rounded corners = 2, color = red] (6,4) to (6,3.5) to (5, 3.5) to (5,4) {};
\draw [rounded corners = 2, color = red] (9,4) to (9,3.5) to (8, 3.5) to (8,4) {};
\draw [rounded corners = 2, color = red] (11,4) to (11,3.5) to (10, 3.5) to (10,4) {};
\draw [rounded corners = 2, color = red] (13,4) to (13,3.5) to (12, 3.5) to (12,4) {};

\end{scope}

\node at (7, 4.6) {\tiny $ = 4M$};
\node at (7, 5.15) {\tiny $ = 12M$};

\node at (1.5, 3.3) {\tiny $= 3$};
\node at (3.5, 3.3) {\tiny $= 3$};
\node at (5.5, 3.3) {\tiny $= 3$};
\node at (8.5, 3.3) {\tiny $= 3$};
\node at (10.5, 3.3) {\tiny $= 3$};
\node at (12.5, 3.3) {\tiny $= 3$};

\foreach \x in {0, 1, ..., 14}
\node[vertex] at (\x, 4) {};

\foreach \x in {1.5, 3.5, 5.5}
\node[vertex] at (\x, 4) {};

\foreach \x in {8.5, 10.5, 12.5}
\node[vertex] at (\x, 4) {};

\foreach \x in {0, 7, 14} 
\node[ rectangle, inner sep=2pt, draw=black, fill=green!50] at (\x, 4) {};



\node at (1.25, 3.8) {\tiny $-1$};
\node at (1.75, 3.8) {\tiny$-1$};
\node at (3.25, 3.8) {\tiny$-1$};
\node at (3.75, 3.8) {\tiny$-1$};
\node at (5.25, 3.8) {\tiny$-1$};
\node at (5.75, 3.8) {\tiny$-1$};

\node at (8.25, 3.8) {\tiny$-1$};
\node at (8.75, 3.8) {\tiny$-1$};
\node at (10.25, 3.8) {\tiny$-1$};
\node at (10.75, 3.8) {\tiny$-1$};
\node at (12.25, 3.8) {\tiny$-1$};
\node at (12.75, 3.8) {\tiny$-1$};

\end{scope}


\node at (0.5, 2.3) {\scriptsize \textbf{Zone} $Z'$};
\foreach \x in {0.5, 2.5, 6.5, 7.5, 11.5, 13.5}
\node at (\x, 1) {\tiny$2M - 3$};


\begin{scope}[thick]

\draw [->, bend left, color = red] (1.5,1) to (1.05, 1) {};
\draw [->, bend left, color = red] (2,1) to (1.55, 1) {};
\draw [->, bend left, color = red] (3.5,1) to (3.05, 1) {};
\draw [->, bend left, color = red] (4,1) to (3.55, 1) {};
\draw [->, bend left, color = red] (5.5,1) to (5.05, 1) {};
\draw [->, bend left, color = red] (6,1) to (5.55, 1) {};
\draw [->, bend left, color = red] (8.5,1) to (8.05, 1) {};
\draw [->, bend left, color = red] (9,1) to (8.55, 1) {};
\draw [->, bend left, color = red] (10.5,1) to (10.05, 1) {};
\draw [->, bend left, color = red] (11,1) to (10.55, 1) {};
\draw [->, bend left, color = red] (12.5,1) to (12.05, 1) {};
\draw [->, bend left, color = red] (13,1) to (12.55, 1) {};

\draw [rounded corners = 2, color = red] (2,1) to (2,0.5) to (1, 0.5) to (1,1) {};
\draw [rounded corners = 2, color = red] (4,1) to (4,0.5) to (3, 0.5) to (3,1) {};
\draw [rounded corners = 2, color = red] (6,1) to (6,0.5) to (5, 0.5) to (5,1) {};
\draw [rounded corners = 2, color = red] (9,1) to (9,0.5) to (8, 0.5) to (8,1) {};
\draw [rounded corners = 2, color = red] (11,1) to (11,0.5) to (10, 0.5) to (10,1) {};
\draw [rounded corners = 2, color = red] (13,1) to (13,0.5) to (12, 0.5) to (12,1) {};

\draw [->, rounded corners = 5, color = blue!50] (1.5,1) to (1.5,2.3) to (5,2.3) to (5,1.05);
\draw [->, rounded corners = 3, color = blue!50] (3.5,1) to (3.5,2.05) to (5,2.05) to (5,1.05);
\draw [->, rounded corners = 3, color = blue!50] (4,1) to (4,1.8) to (5.5,1.8) to (5.5,1.05);
\draw [->, rounded corners = 3, color = blue!50] (8.5,1) to (8.5,2.3) to (10,2.3) to (10,1.05);
\draw [->, rounded corners = 3, color = blue!50] (9,1) to (9,1.8) to (10.5,1.8) to (10.5,1.05);
\draw [->, rounded corners = 3, color = blue!50] (9,1) to (9,2.05) to (12.5,2.05) to (12.5,1.05);

\draw [->, rounded corners = 5] (14,1) to (14,0) to (0,0) to (0,0.9);

\end{scope}

\draw [->, thick, bend left, gray] (4,1) to (4.95,1) {};
\draw [->, thick, bend left, gray] (5,1) to (4.05,1) {};
\draw [->, thick, bend left, gray] (9,1) to (9.95,1) {};
\draw [->, thick, bend left, gray] (10,1) to (9.05,1) {};

\node at (3,2.4) {\tiny $ 4M + 2$};
\node at (3,1.75) {\tiny $ 2M + 2$};
\node at (6,1.6) {\tiny $2M + 2$};

\node at (9.1,2.4) {\tiny $2M + 2$};
\node at (10.8,2.25) {\tiny $ 4M + 2$};
\node at (11,1.5) {\tiny $ 2M + 2$};

\node at (4.5, 1.25) {\tiny $2M + 1$};
\node at (4.5, 0.75) {\tiny $- 2M$};
\node at (9.5, 1.25) {\tiny $2M + 1$};
\node at (9.5, 0.75) {\tiny $-2M$};

\node at (1.5, 0.3) {\tiny $= 3$};
\node at (3.5, 0.3) {\tiny $= 3$};
\node at (5.5, 0.3) {\tiny $= 3$};
\node at (8.5, 0.3) {\tiny $= 3$};
\node at (10.5, 0.3) {\tiny $= 3$};
\node at (12.5, 0.3) {\tiny $= 3$};

\node at (1.25, 0.8) {\tiny $-1$};
\node at (1.75, 0.8) {\tiny$-1$};
\node at (3.25, 0.8) {\tiny$-1$};
\node at (3.75, 0.8) {\tiny$-1$};
\node at (5.25, 0.8) {\tiny$-1$};
\node at (5.75, 0.8) {\tiny$-1$};

\node at (8.25, 0.8) {\tiny$-1$};
\node at (8.75, 0.8) {\tiny$-1$};
\node at (10.25, 0.8) {\tiny$-1$};
\node at (10.75, 0.8) {\tiny$-1$};
\node at (12.25, 0.8) {\tiny$-1$};
\node at (12.75, 0.8) {\tiny$-1$};

\node at (7,-0.15) {\tiny $ - (16M + 1)$};

\foreach \x in {0, 1, ..., 14}
\node[vertex] at (\x, 1) {};

\foreach \x in {1.5, 3.5, 5.5}
\node[vertex] at (\x, 1) {};

\foreach \x in {8.5, 10.5, 12.5}
\node[vertex] at (\x, 1) {};

\foreach \x in {0, 7, 14} 
\node[ rectangle, inner sep=2pt, draw=black, fill=green!50] at (\x, 1)
{};











\end{tikzpicture}
\caption{Illustration of the zone $Z$ and $Z'$ for the formula $(p_1
  \lor p_2 \lor \neg p_3) \land (p_3 \lor \neg p_4 \lor \neg p_1)$. The
  separator clocks $r_0, r_1, r_2$ are shown by the green boxes
  (leftmost box is $r_0$, middle one is $r_1$ and the rightmost is
  $r_2$). The intermediate literal clocks are shown by the black dots:
  between $r_0$ and $r_1$ are $x_1^1, y_1^1, z_1^1, x_1^2,
  y_1^2,z_1^2, x_1^3, y_1^3, z_1^3$ in the same sequence. Similarly
  between $r_1$ and $r_2$ are the clocks $x_2^1, \dots, z_2^3$.  An edge of the
  form $x \xrightarrow{c} y$ simply denotes the constraint $y - x \leq c$, 
  whereas edges $x \xrightarrow{{}=~c} y$ mean that
  $y- x = c$.  When we write $c$ between two consecutive clocks, we mean that the
  difference between them equals $c$.}
\label{fig:hardness}
\end{figure}
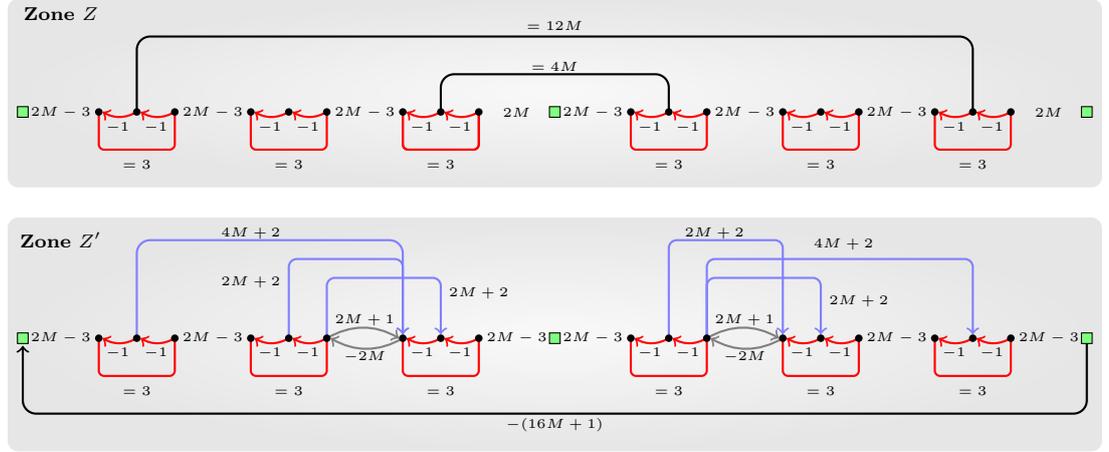

Constraints (\ref{eq:within-a-block}) and (\ref{eq:between-blocks})
ensure that for every valuation in $Z$ we have the following order of
clocks for each $i \in \{1, \dots, N\}$:
\begin{align}\label{eq:order-in-Z}
  r_{i-1} ~~< ~~x_i^1 < y_i^1 < z_i^1 ~~<~~ x_i^2 < y_i^2 < z_i^2
  ~~<~~ x_i^3 < y_i^3 < z_i^3 ~~<~~ r_i
\end{align}
In every valuation of $Z$, we have $x_i^j \tight y_i^j \tight z_i^j$
for every literal $l_{i}^j$. This is because we have assumed that $M >
3$ and we have restricted the gaps between $x_i^j, y_i^j$ and $y_i^j,
z_i^j$ to be in the interval $[1,2]$ (c.f.
(\ref{eq:within-a-block})).  We do not want any other pair of clocks
that are consecutive according to the above ordering to be
tight. Hence we choose the rest of the gaps to be strictly more than
$M$ (c.f.~(\ref{eq:between-blocks})). This gives a $\tight^*$ picture
in which each $\{x_i^j, y_i^j, z_i^j\}$ forms a block of ``length''
$3$ and the gaps between each such blocks, or between a block and a
separator is larger than $M$. Note that for each $v \in Z$, we also
have $v(r_i) - v(r_{i-1}) = 8M$ for $i \in \{1, \dots, N\}$.  We will
next enforce that literals in $\varphi$ involving the same variable
have the same $y - x$ and $z - y$ values for their corresponding
clocks. Without loss of generality, we assume that the three literals
corresponding to the same clause have different variables. Therefore
this condition is relevant for literals in different clauses, but with
the same variable. For every $l_i^j$ and $l_{i'}^{j'}$ such that
$\var(l_i^j) = \var(l_{i'}^{j'})$ and $i' > i$:
\begin{align}
  \label{eq:same-literal-same-difference}
  y_{i'}^{j'} - y_i^j = (i'-i)\cdot 8M + (j'-j)\cdot 2M
\end{align}
Note that from (\ref{eq:within-a-block}) and (\ref{eq:between-blocks})
we can infer that the values of $v(x_{i'}^{j'}) - v(x_i^j)$ and
$v(z_{i'}^{j'}) - v(z_i^j)$ are already equal to the right hand side of the
above equation, as the $x$ and $z$ clocks are ``fixed'' and $y$ is
``flexible''. Constraint (\ref{eq:same-literal-same-difference}) then
ensures that $v(y_i^j) - v(x_i^j) = v(y_{i'}^{j'}) - v(x_{i'}^{j'})$
and $v(z_i^j) - v(y_i^j) = v(z_{i'}^{j'}) - v(y_{i'}^{j'})$ whenever
$l_i^j$ and $l_{i'}^{j'}$ with $i' > i$, have the same variable.

\medskip \textbf{Encoding of assignments:} \label{encoding} We call a
valuation $v$ to be \emph{integer tight} if for every pair of clocks
$x,y$ such that $x \tight^* y$ in $v$, we have $v(y) - v(x)$ to be an
integer. By construction of $Z$, a valuation $ v \in Z$ will be
integer tight if $v(y_i^j) - v(x_i^j)$ is an integer for every $i,
j$. Moreover, by construction, this value can either be $1$ or $2$. We
will use such integer tight valuations to encode the variable
assignments. An integer tight valuation $v$ encodes the assignment
$\s_v$ given by: $\s_v(\var(l_i^j)) = \ttt$ if $v(y_i^j) - v(x_i^j) =
1$ and $\s_v(\var(l_i^j)) = \fff$ if $v(y_i^j) - v(x_i^j) = 2$.
By (\ref{eq:same-literal-same-difference}), the above assignment is
well defined. Moreover, the zone $Z$ contains an integer tight
valuation for every possible assignment.


We have encoded assignments to variables using integer tight
valuations. An assignment $\s$ satisfies $\varphi$ if every clause
evaluates to true under $\s$. From a valuation $v$ encoding this
assignment $\s$, we need a mechanism to check whether each clause is
true. This is where we will use the clock differences which are not
tight, that is the ones which are $>M$. Clauses will be identified by
certain kind of shifts to these unbounded differences in $v$. We will
introduce some more notation.  Let $L := \{ (x_i^j, y_i^j, z_i^j)~|~i
\in \{1, \dots, N\} \text{ and } j \in \{1, 2, 3\}\}$ be the triplets
of clocks associated with each literal. A literal is said to be
\emph{positive} if it is a variable $p$, and it is \emph{negative} if
it is the negation $\neg p$ of some variable $p$. We will assume that
in every clause of $\varphi$, the positive literals are written before
the negative literals: for example, we write $p_1 \lor p_3 \lor \neg
p_2$ instead of $p_1 \lor \neg p_2 \lor p_3$. For each clause $C_i$,
let $(e_i,f_i)$ be the pair of clocks corresponding to $C_i$ in the
border between positive and negative literals:
\begin{align}
  \label{eq:border-clocks}
  (e_i,f_i) := \begin{cases}
    (r_{i-1}, x_i^1) & \text{ if all literals in $C_i$ are negative } \\
    (z_i^j, x_i^{j+1}) & \text{ if for $j \in \{1, 2\}$, $l_i^j$ is
      positive and $l_i^{j+1}$ is negative } \\
    (z_i^3, r_i) & \text{ if all literals in $C_i$ are positive }
  \end{cases}
\end{align}
Given the formula $\varphi$, the above border clocks are fixed.  For a
valuation $v \in Z$ and $i\in\{1,\dots,N\}$, define $v_i$ to be the
valuation such that:
\begin{itemize}
\item $v_i(y) - v_i(x) = v(y) - v(x)$ and $v_i(z) - v_i(y) = v(z) -
  v(y)$ for all $(x,y,z) \in L$
\item $v_i(f_i) - v_i(e_i) = 2M + 1$ and $v_i(f_{i'}) - v_i(e_{i'}) =
  2M$ for all $i' \neq i$,
\item $v_i(r_0)=0$ and all other differences between consecutive
  clocks (according to order given by~(\ref{eq:order-in-Z})) is $2M -
  3$.
\end{itemize}
Valuation $v_i$ acts as a representative for the clause $C_i$, through
the choice of the difference $2M + 1$ in the border of $C_i$, and $2M$
in the other borders. We want to construct zone $Z'$ such that when
$C_i$ is true, the valuation $v_i$ forms a negative cycle with the
constraints of $Z'$, via the literal which is true in $C_i$.

\medskip

\textbf{Construction of $Z'$.}  Zone $Z'$ is described by five sets of
constraints. The first set of constraints are between the clocks of
the same literal, and are identical to that in $Z$:
\begin{align}
  \label{eq:within-a-block-Z'}
  y_i^j - x_i^j \geq 1 \quad \text{ and } \quad z_i^j - y_i^j \geq 1
  \quad \text{ and } \quad z_i^j - x_i^j = 3
\end{align}
The second set of constraints are for border clocks in each
clause. For each $i \in \{1, \dots, N\}$:
\begin{align}
  \label{eq:border-difference-Z'}
  2M \le f_i - e_i \le 2M + 1
\end{align}
where $e_i$ and $f_i$ are according to the definition in
(\ref{eq:border-clocks}).  The third set of constraints fix
differences between consecutive blocks not involving border clocks to
$2M - 3$.
\begin{align}
  \label{eq:non-border-Z'}
  x_i^1 - r_{i-1} & = 2M - 3 ~~ \text{ if } (r_{i-1}, x_i^1) \neq
  (e_i,
  f_i)  \text{ and } \\
  \nonumber x_i^{j+1} - z_i^j & = 2M - 3 ~~\text{ for } j \in \{1, 2\}
  \text{ when } (z_i^j, x_{i}^{j+1})
  \neq (e_i, f_i) ~~\text{ and } \\
  \nonumber r_{i} - z_i^3 & = 2M - 3 ~~\text{ when } (z_i^3, r_i) \neq
  (e_i, f_i)
\end{align}

From
(\ref{eq:within-a-block-Z'},\ref{eq:border-difference-Z'},\ref{eq:non-border-Z'}),
we see that for every valuation in $Z'$ the difference between
separators, that is $r_i - r_{i-1}$, is between $8M$ and $8M + 1$ with
the flexibility coming due to $f_i - e_i$. The fourth set of
constraints ensures that at least one of the $f_i - e_i$ should be
bigger than $2M$.
\begin{align}
  \label{eq:choose-clause-Z'}
  r_N - r_0 \ge (8M\cdot N) + 1
\end{align}

So far, the constraints that we have chosen for $Z'$ do not talk about
clauses being true or false. 
%
Recall that valuation $v_i$ where the border $v_i(f_i) - v_i(e_i) =
2M+1$ represents the choice of $C_i$ for evaluation.
The final set of constraints ensure that for every valuation $v'$ in
$Z'$ which has integer values for the $y - x$ values and has $v'(f_i)
- v'(e_i) = 2M + 1$, every literal in $C_i$ evaluates to false under
the encoding scheme given in Page~\pageref{encoding}: that is, if
$l_i^j$ is positive then $v'(y_i^j) - v'(x_i^j)$ cannot be $1$ and
when $l_i^j$ is negative, $v'(y_i^j) - v'(x_i^j)$ cannot be $2$. For a
positive literal $l_i^j$ let $d_i^j\in\{0,1,2\}$ be the number of
$(x,y,z)$ blocks corresponding to positive literals between $z_i^j$
and $e_i$ (does not include $j$). Similarly, for a negative literal,
let $d_i^j\in\{0,1,2\}$ be the number of blocks corresponding to
negative literals between $f_i$ and $x_i^j$ (again, excludes $j$). We
add the following constraints: 
\begin{align}
  \label{eq:true-implies-constraint-not-satisfied-in-Z'}
  f_i - y_i^j & \le d_i^j \cdot 2M + (2M + 2) \quad \text{ if } l_i^j
  \text{ is a
    positive literal } \\
  \nonumber y_i^j - e_i & \le d_i^j \cdot 2M + (2M + 2) \quad \text{
    if } l_i^j \text{ is a negative literal }
\end{align}

\begin{theorem}\label{thm:hardness}
  Formula $\varphi$ is satisfiable iff $Z \not \fleq_M Z'$. The
  decision problem $Z \not \fleq_M Z'$ is $\NP$-hard.
\end{theorem}
\begin{proof}(Sketch.)  Assume $\varphi$ is satisfiable. Consider the
  valuation $v \in Z$ corresponding to the satisfying assignment. Pick
  an arbitrary $v' \regeq_M v$. If $v'$ were to lie in $Z'$, by
  (\ref{eq:choose-clause-Z'}), at least one of the border differences
  should be $> 2M$. This forms a contradiction with the literal
  that is true in clause $C_i$ due to
  (\ref{eq:true-implies-constraint-not-satisfied-in-Z'}).

  Assume $Z \not \fleq_M Z'$. As $Z$ and $Z'$ are topologically
  closed, and the $\tight^*$ equivalence classes are same for every
  valuation in $Z$, by
  Proposition~\ref{prop:topo-closed-closure-inclusion} there is an
  integral valuation $v$ such that $\reg{v} \cap Z'$ is empty. This
  $v$  gives a satisfying assignment: mainly, each $v_i$ corresponding to $v$
  will form a negative cycle with some literal clocks of $C_i$, and
  this literal will be made true by the assignment corresponding to $v$.
\end{proof}

\begin{theorem}
  The decision problem $Z \not \lud Z'$ is NP-hard.

\end{theorem}

Since $\fleq_M$ is just a special case of $\lud$, this result follows.


\section{Conclusion}
\label{sec:exper-concl}

In this paper, we have proposed a simulation $\lud$ and a simulation
test $Z \lud Z'$ that facilitates a forward analysis procedure for
timed automata with diagonal constraints. An abstraction function
  on based on $\md$ was already proposed
  in~\cite{Bouyer:2004:forwardanalysis} in the context of forward
  analysis using explicit abstractions, but it was not used as no
  efficient storage mechanisms for non-convex abstractions are known.
Moreover, no simulation test apart from a brute force check of
enumerating over all regions was known either. Here, we provide a more
refined simulation test, which in principle gives a more structured way
of performing this enumeration. In the diagonal free case, this turns
out to be $O(|X|^2)$~\cite{Herbreteau:IandC:2016}. But, as we show
here, in the presence of diagonal constraints, $Z \not \lud Z'$ is
$\NP$-complete. Nevertheless, having this forward analysis framework
creates the possibility to incorporate recent optimizations studied
for diagonal free automata which crucially depend on this inclusion
test, and have been indispensable in improving the performance
substantially~\cite{Herbreteau:2011:FSTTCS,
  Herbreteau:2013:CAV}. Moreover, we believe that this framework can
be extended to various other problems involving timed automata with
diagonal constraints, for instance liveness verification and cost
optimal reachability in priced timed automata.

\renewcommand{\arraystretch}{1.4}
\begin{table}[t]
\scriptsize
\centering
\begin{tabular}{|l|c|c|c|c|c|}
  \hline
  \multicolumn{2}{|c|}{Model} & \multicolumn{2}{|c|}{Diagonal
    constraints + $\alud$} & \multicolumn{2}{|c|}{Diagonal free + $\alu$}
 \\
 \hline
 Name & \# clocks & \# zones & time (in sec.) & \# zones &  time (in
 sec.) \\
 \hline
 Cex 1 & 4 & 8 & 0.07 & 22 & 0.05\\
 \hline 
 Cex 2 & 8 & 437 & 123 & 2051 & 0.11 \\
 \hline
 Fischer 4 & 8 & 2618 & 170 & 73677 & 2.1\\
 \hline
 Fischer 5 & 10 & 15947 & 2170 & 1926991 & 134 \\
 \hline
\end{tabular}
\caption{Experiments to compare forward analysis with diagonal
  constraints versus forward analysis on the equivalent diagonal free
  automaton. ``Fischer K''
  is a model of a communication protocol with K processes
  as described in \cite{Reynier:toolreport}. Cex 1 is the automaton in
  \cite{Bouyer:2004:forwardanalysis} which revealed the bug with 
  the explicit abstraction method. Cex 2 is a similar version
  with more states,
  given in \cite{Reynier:toolreport}. }
\label{tab:experiments}
\end{table}

We have implemented the simulation test $Z \lud Z'$ in a prototype tool
T-Checker~\cite{tchecker} which has been developed for diagonal free
timed automata. The simulation test constructs an SMT formula in linear
arithmetic and invokes the Z3 solver~\cite{z3}. Preliminary
experiments on models from \cite{Reynier:toolreport} are reported in
Table~\ref{tab:experiments}. For each model $\Aa$ (with diagonal
constraints), the table compares the performance of running the
forward analysis approach using $\lud$ on $\Aa$ (Columns 3 and 4)
versus the forward analysis using (diagonal free variant) $\alu$
\cite{Behrmann:STTT:2006} on the equivalent diagonal free automaton
$\Aa_{df}$ (Columns 5 and 6). We observe that there is a significant
decrease in the number of nodes explored while using $\lud$ on
$\Aa$. The problem with $\Aa_{df}$ is that each state $q$ of $\Aa$ has
$2^d$ copies in $\Aa_{df}$ if $d$ is the number of diagonal
constraints (essentially, the states of $\Aa_{df}$ maintain the
information about whether each diagonal is true or false when reaching
this state). Therefore a simulation of the form $ Z \lud  Z')$ arising from
$(q, Z)$ and $(q, Z')$ which
occurs in the analysis of $\Aa$ might not be possible while
analyzing $\Aa_{df}$ just because the corresponding paths reach
different copies of $q$, say $(q_1, Z)$ and $(q_2, Z')$. This prunes
the search faster in $\Aa$. Indeed, exploiting the conciseness of
diagonal constraints could be a valuable tool for modeling and
verifying real-time systems. On the other hand, we note that our
algorithm performs bad in terms of timing due to the costlier
simulation test, even while there is a good reduction in the number of
nodes. Given this decrease in the number of nodes, it is interesting
to investigate efficient methods for $Z \lud Z'$ by making best use of
the SMT solver. This, and comparing our method with other
approaches~\cite{Bouyer:2005:diagonal-refinement} is part of future
work.


\bibliographystyle{plainurl}
\bibliography{main}

\newpage

\appendix

\section{Appendix for Section~\ref{sec:hardness}}

\label{app:hardness}

\subsection{An observation about closed zones}

The aim of this subsection is to prove the following proposition.

\medskip

\Repeat{Proposition}{prop:topo-closed-closure-inclusion}
  Let $Z$ be a topologically closed zone such that the
  $\tight^*$ equivalence classes of every valuation in $Z$ are the
  same. Let $Z'$ be a zone with $Z \not\md Z'$. Then,
  there exists an integral valuation $u \in Z$ such that $\reg{u} \cap
  Z'$ is empty.

\medskip

We will state below some intermediate lemmas before proving this
proposition.

\begin{lemma}\label{lem:topo-closed-integral-valuation}
  Let $Z$ be a topologically closed zone. Then, $Z$ is non-empty iff
  it contains an integral valuation.
\end{lemma}
\begin{proof}
  If $Z$ contains an integral valuation, then clearly $Z$ is
  non-empty. Let us prove the converse.  Assume $Z$ is non-empty. Pick
  a valuation $v \in Z$. Consider a new valuation $v'$ defined as
  $v'(x) = \lfloor v(x) \rfloor$, for all clocks $x$.  Note that $v'$ is an integral
  valuation. We will show that $v'$ satisfies each constraint $x - y
  \le c$ in $Z$.  We write $\{ v(x)
  \}$ for the fractional part of $v(x)$. We already know that $v(x) -
  v(y) \le c$. This implies that $(\lfloor v(x) \rfloor - \lfloor v(y)
  \rfloor) + (\{v(x)\} - \{v(y)\}) \le c$. Using the definition of
  $v'$, we get that $v'(x) - v'(y) \le c + \{v(y)\} - \{v(x)\}$. Since
  $\{v(y)\} - \{v(x)\} < 1$ and $v'(x) - v'(y)$ is an integer, we get
  that $v'(x) - v'(y) \le c$.
\end{proof}

Let $G_Z$ be the canonical distance graph of a zone, and let $E$ be a
set of its edges. Define:
\begin{align*}
  \MinSum_Z(E) ~:=~ \min_{v \in Z} \sum\limits_{y \rightarrow x \in E}
  v(x) - v(y)
\end{align*}

The following lemma claims that the above minimum sum is attained by
an integral valuation if $Z$ is topologically closed, and the set of
edges $E$ satisfy a particular property.

\begin{lemma}
  \label{lem:topo-closed-minsum-integral}
  Let $Z$ be a non-empty topologically closed zone. Let $E$ be a set
  of edges in the canonical distance graph $G_Z$ such that no two
  edges in $E$ have common vertices. Moreover, for every edge $y \to
  x$ in $E$, the weight of the reverse edge $x \to y$ is not $(<,
  \infty)$ in $G_Z$. Then, there exists an integral
  valuation $v$ such that $\MinSum_Z(E)$ equals $\sum\limits_{y \to x
    \in E} v(x) - v(y)$.
\end{lemma}
\begin{proof}
  Let $G_Z$ be the canonical distance graph representing $Z$. Suppose
  $E$ is $\{y_1 \to x_1, y_2 \to x_2, \dots, y_k \to x_k \}$. Denote
  by $E_x$ the set $\{x_1, \dots, x_k\}$ and by $E_y$ the set $\{y_1,
  \dots, y_k\}$. By assumption on $E$ that no two edges intersect, we
  get that variables $x_1, \dots, x_k, y_1, \dots, y_k$ are pairwise distinct. Note that
  $\MinSum_Z(E)$ can be rewritten:
  \begin{align*}
    \MinSum_Z(E) = \min_{u \in Z}~\left[~\left( \sum_{x \in E_x} u(x)
      \right) - \left( \sum_{y \in E_y} u(y) \right)~\right]
  \end{align*}
  For $i, j \in \{1, \dots, k\}$ let the weight of the edge $x_i \xra{} y_j$ in
  $G_Z$ be $(\le, c_{ij})$ or $(<,\infty)$ (by assumption we know that
  the weight of $x_i \xra{} y_i$ is not $(<,\infty)$ and hence
  $c_{ii}\neq\infty$).  This implies the constraint $y_j - x_i \le c_{ij}$,
  or seen in a different way: $x_i - y_j \ge -c_{ij}$.

  Let $\mathfrak{S}_k$ denote the set of all permutations of $\{1,
  2, \dots, k\}$.  Each pair of permutations $\pi, \pi' \in \mathfrak{S}_k$
  gives a permutation of $E_x$ and $E_y$ and hence fixes
  a collection of $k$ edges of the form:
  \begin{align*}
    x_{\pi'(i)} \to y_{\pi(i)} \quad \text{ with weight } (\le,
    c_{\pi'(i)\pi(i)}) \text{ or } (<, \infty)
  \end{align*}
  Note that in the above, weights of the form $(<, \infty)$ are
  possible, since we have only guaranteed that for the identity
  permutation, the weights are finite. Call a pair of
  permutations $\pi, \pi'$ to be finite if none of its associated edges
  is $(<, \infty)$. Each finite pair of permutations gives the
  following constraint satisfied by every valuation in the zone:
  \begin{align}
    \label{eq:1}
    \left(~\sum\limits_{i=1}^{i=k} y_{\pi(i)} - x_{\pi'(i)}~ \right)
    ~~\le~~\sum\limits_{i=1}^{i=k} c_{\pi'(i)\pi(i)}
  \end{align}
  Call the sum on the right hand side as $c_{\pi',\pi}$. Rewriting the
  above equation gives the following constraint:
  \begin{align*}
    \left( \sum_{y \in E_y} y \right) - \left( \sum_{x \in E_x} x
    \right) ~~\le~~c_{\pi'\pi}
  \end{align*}
  Since this is true for every finite pair of permutations $\pi,\pi'$, we get
  the following constraint from $G_Z$:
  \begin{align*}
    \left( \sum_{y \in E_y} y \right) - \left( \sum_{x \in E_x} x
    \right) ~~\le~~\min_{\pi, \pi' \text{ finite }} c_{\pi'\pi}
  \end{align*}
  Let $c^*$ the minimum value given by the right hand side of the
  above equation. As the identity permutation is finite, we ensure
  that $c^*$ will be a finite value. Moreover, since this is a contraint obtained from $G_Z$,
  every valuation $u \in Z$ satisfies it. This gives:
  \begin{align*}
    \left( \sum\limits_{x\in E_x} u(x) \right) - \left( \sum\limits_{y
        \in E_y} u(y) \right) ~~ \ge~~ -c^*
  \end{align*}
  Hence we get that $\MinSum_Z(E) \ge -c^*$.  We claim that there
  exists an integral valuation $v \in Z$ for which the associated sum
  attains the value $-c^*$. This will prove the lemma.

  Assume that the minimum value $c^*$ is obtained with permutations
  $\pi, \pi'$. Consider a distance graph $G_{Z_1}$ obtained from $G_Z$
  by setting the $y_{\pi(i)} \to x_{\pi'(i)}$ edge to $(\le,
  -c_{\pi'(i)\pi(i)})$ for all $i \in \{1, \dots, k\}$ and keeping the
  rest of the edges same as in $G_Z$. This gives a zero cycle
  $y_{\pi(i)} \to x_{\pi'(i)} \to y_{\pi(i)}$ and amounts to saying
  that every valuation in $\sem{G_{Z_1}}$ has $x_{\pi'(i)} -
  y_{\pi(i)} = -c_{\pi'(i)\pi(i)}$ for every $i$. This means that for every
  valuation  $v_1 \in G_{Z_1}$, we have:
  \begin{align*}
     ~\sum\limits_{i=1}^{i=k} \left(~ v_1(x_{\pi'(i)}) -
      v_1(y_{\pi(i)}) ~\right) &  = -c^*
    \\
\imp  ~ \sum_{x \in E_x} v_1(x) - \sum_{y \in E_y} v_1(y) &  = - c^*
  \end{align*}
  Moreover, note that  $G_{Z_1}$ represents a topologically
  closed zone, since all weights are $(\le,c)$ or $(<, \infty)$
  (which are inherited from $G_Z$). If $\sem{G_{Z_1}}$ is non-empty, we can
  employ Lemma \ref{lem:topo-closed-integral-valuation} to say that
  there exists an integral valuation in $\sem{G_{Z_1}}$ which from the
  above discussion would attain the minimum sum.

  We will now show that $\sem{G_{Z_1}}$ is indeed non-empty, for which
  it is sufficient to show that there
  are no negative cycles in $G_{Z_1}$. Since we started with a
  non-empty zone $G_Z$, any negative cycle in $G_{Z_1}$ would be due
  to the modified edges. Colour all the modified edges $y_{\pi(i)} \to
  x_{\pi'(i)}$ by red, and make the rest of the edges green. Note that
  the sum of all the red edges gives $-c^*$.

  Suppose $G_{Z_1}$ has a negative cycle $C$.  Two consecutive edges in $C$
  cannot be coloured red as this would contradict the fact that no two edges in
  $E$ have common vertices (all clocks $x_1, \dots, x_k, y_1, \dots, y_k$ are
  distinct).  Since $G_Z$ is canonical, a maximal path of green edges in $C$ can
  be replaced with a single green edge $r \to t$ in $G_Z$ from its source $r$ to
  its target $t$ (notice that by maximality, $r\to t$ cannot be red).
  Therefore, we can assume that the negative cycle $C$ consists of
  alternating red and green edges and takes the following form (green edges are
  shown as $\to$ and red edges are shown using $\redarrow$).
  \begin{align*}
    C: \quad y_{\pi(i_1)} \redarrow x_{\pi'(i_1)} \to y_{\pi(i_2)}
    \redarrow x_{\pi'(i_2)} \to \cdots \to y_{\pi(i_m)} \redarrow
    x_{\pi'(i_m)} \to y_{\pi(i_1)}
  \end{align*}
  Recall that the weight of a red arrow $y_{\pi(i_p)} \redarrow
  x_{\pi'(i_p)}$ is $(\le, -c_{\pi'(i_p)\pi(i_p)})$ and the weight of
  a green arrow $x_{\pi'(i_p)} \to y_{\pi(i_{p+1})}$ is $(\le,
  c_{\pi'(i_p) \pi(i_{p+1})})$. The fact that $C$ is a negative cycle
  implies:
  \begin{align*}
    c_{\pi'(i_1) \pi(i_2)} + c_{\pi'(i_2)\pi(i_3)} + \dots +
    c_{\pi'(i_m)\pi(i_1)} < c_{\pi'(i_1)\pi(i_1)} + c_{\pi'(i_2)
      \pi(i_2)} + \dots + c_{\pi'(i_m) \pi(i_m)}
  \end{align*}
  Define permutation $\pi_1$ such that $\pi_1(i_1) = \pi(i_2),
  \pi_1(i_2)= \pi(i_3), \dots, \pi_1(i_m) = \pi(i_1)$ and $\pi_1(j) =
  \pi(j)$ for $j \notin \{i_1, \dots, i_m\}$. The above equation
  suggests that $\sum_{i=1}^{i=k} c_{\pi'(i)\pi_1(i)}$ is strictly
  smaller than $\sum_{i=1}^{i=k} c_{\pi'(i)\pi(i)}$. This contradicts
  that $\pi',\pi$ gave the minimum sum. Hence there cannot be a
  negative cycle in $G_{Z_1}$.
\end{proof}

Let $\gvm$ be the distance graph for $\reg{v}$ as defined in Definition \ref{def:gvlu-distance-graph}.
Then every edge $y \to x$ in $\gvm$ has weight either $(<,-M)$ (when $v(x) - v(y) < -M$) or $(\le, v(x) - v(y))$ (when $-M \le v(x) - v(y) \le M$, that is when $y \tight x$).
In $\gvm$ if for some pair of clocks $x,y$ there is no edge $y \to x$ and $y \tight^* x$, add an edge $y \to x$ with weight $(\le, v(x) - v(y))$.
Let us call this new graph $\gv$.

%
%
%
%
%
\begin{proof}[Proof of Proposition~\ref{prop:topo-closed-closure-inclusion}]
  Since $Z
  \not \md Z'$ there exists a $v \in Z$ such that $\reg{v}
  \cap Z'$ is empty.
  From Theorem \ref{thm:alu-inclusion} we have that 
  $\min(\gvm, G_{Z'})$ contains a negative cycle in which no two consecutive edges are from $G_{Z'}$.
  Since all the edges that are present in $\gvm$ are also present in $\gv$, this negative cycle in $\min(\gvm, G_{Z'})$ will also be present in $\min(\gv, G_{Z'})$, let us call this cycle $N^*$.
  We can replace every sequence of consecutive edges $\{ x_1 \to x_2, x_2 \to x_3, \dots, x_{k-1} \to x_k \}$ in $N^*$, where each edge $x_i \to x_j$ has weight $(\le, v(x_j) - v(x_i))$, with an edge $x_1 \to x_k$ with weight $(\le, v(x_k) - v(x_1))$.
  This new edge is also from $\gv$ since $x_1 \tight^* x_k$. Hence this modified cycle, call it $N$, would be a part of $\min(\gv, G_{Z'})$ and would still be negative.
  We can then assume that no two
  tight edges (that is, $y \to x$ such that
  $y \tight^* x$) coming from $\gv$ in $N$ are consecutive.
  We first colour the edges:
  \begin{description}
  \item[Red:] all edges from $G_{Z'}$,
  \item[Yellow:] all $y \to x$ edges from $\gv$ that have $y \tight^*
    x$, we will denote them by $x \dra y$,
  \item[Blue:] all $y \to x$ edges from $\gv$ with weight $(<, -M)$.
  \end{description}
  Let the sum of the weights of the red, yellow and blue edges be
  $W_{red}$, $W_{yellow}$ and $W_{blue}$ respectively. We know that
  $W_{red} + W_{yellow} + W_{blue}$ is strictly less than $(\le, 0)$.

  Let $E$ be the set of yellow edges.
  Since no two consecutive edges
  are yellow, we get that no two edges in $E$ have
  common vertices. For every edge $y \yellowedge x$ in $E$,
  as $y \tight^* x$, we can infer that $x \to y$ will be a finite edge
  in $G_Z$. We can then use Lemma
  \ref{lem:topo-closed-minsum-integral} to get an integral valuation
  $u$ in which the sum of weights of the yellow edges has a value
  $W'_{yellow} \le W_{yellow}$. Morever, as we have assumed that every
  valuation in $Z$ has the same $\tight^*$ equivalence classes, each
  yellow edge $y \dra x$ in $\gu$ will have weight $(\le, u(x) -
  u(y))$. Therefore, the sum of
  the weights of the yellow edges in $\gu$ will be
  $W'_{yellow}$. Similarly, each blue edge $y \to x$ will be $(<,
  -M)$ in $\gu$, and hence the
  sum of the weights of all blue edges in $\gu$ would be the same
  $W_{blue}$. This shows that the cycle in $\min(\gu, G_{Z'})$ given
  by the edges of $N$ will have value $W_{red} +
  W'_{yellow} + W_{blue}$ which will be negative. This proves that
  $\reg{u} \cap Z'$ is empty.
\end{proof}

\subsection{Reduction from SAT}

We elaborate the proof of the following theorem.

\medskip

\Repeat{Theorem}{thm:hardness} Formula $\varphi$ is satisfiable iff $Z
\not \md Z'$. The decision problem $Z \not \md Z'$
is $\NP$-hard.

\medskip

\begin{lemma}\label{lem:gaps-between-r's} For every $u'$ satisfying
(\ref{eq:within-a-block-Z'}, \ref{eq:border-difference-Z'},
\ref{eq:non-border-Z'}), we have $8M \leq u'(r_i) - u'(r_{i-1}) \leq
8M + 1$ for $i \in \{1, \dots, N\}$.
\end{lemma}
\begin{proof}
  Pick a $u'$ satisfying (\ref{eq:within-a-block-Z'},
  \ref{eq:border-difference-Z'}, \ref{eq:non-border-Z'}). We have
  $u'(r_i) - u'(r_{i-1})$ to be equal to:
  \begin{align*}
    & \qquad  u'(r_i) - u'(z_i^3) ~+~ u'(z_i^3) - u'(x_i^3) ~+~ u'(x_i^3) - u'(z_i^2) ~+ \\
    & u'(z_i^2) - u'(x_i^2) ~+~ u'(x_i^2) - u'(z_i^1) ~+~ u'(z_i^1) -
    u'(x_i^1) ~+~ u'(x_i^1) - u'(r_{i-1})
  \end{align*}
  From (\ref{eq:within-a-block-Z'}), we know that $v'(z_i^j) -
  v'(x_i^{j})$ is $3$. Among the other $4$ clock differences, we know
  that exactly one difference indicates the border, and from
  (\ref{eq:border-difference-Z'}) it lies between $2M$ and $2M +
  1$. The rest are $2M - 3$ due to (\ref{eq:non-border-Z'}). This
  shows that $8M \leq v(r_i) - v(r_{i-1}) \leq 8M + 1$.
\end{proof}

\begin{lemma}\label{lem:sat-diag} Let $v \in Z$ be an integer tight
valuation such that $\s_v \models \varphi$. Then, $v' \not \in Z'$ for
every $v'$ satisfying $v' \regeq_M v$.
\end{lemma}
\begin{proof}
  Pick a $v'$ that is $\regeq_M$ equivalent to $v$. Note that by
  definition of $\regeq_M$, the tight differences in $v$ remain the
  same in $v'$, and the non-tight differences are $> M$. Therefore
  $v'$ satisfies (\ref{eq:within-a-block-Z'}). If $v'$ does not
  satisfy either (\ref{eq:border-difference-Z'}),
  (\ref{eq:non-border-Z'}) or (\ref{eq:choose-clause-Z'}), then $v'
  \not \in Z'$ and we are done.  Otherwise, we have a $v'$ satisfying
  all these constraints.
  By Lemma~\ref{lem:gaps-between-r's}, we have $8M \le v'(r_i) - v'(r_{i-1})
  \le 8M + 1$ for all $i\in\{1,\dots,N\}$.

  From the previous assumption that $v'$ satisfies
  (\ref{eq:choose-clause-Z'}), we get that there is some $i$ for which
  $v'(r_i) - v'(r_{i-1}) > 8M$.  Let us now fix this $i$.  As $v'$
  satisfies (\ref{eq:within-a-block-Z'},\ref{eq:border-difference-Z'},\ref{eq:non-border-Z'}),
  we will have that
  $v'(f_i) - v'(e_{i}) > 2M$.

  Since $v$ is an integer tight valuation such that $\s_v \models
  \varphi$, some literal $l_i^j$ should evaluate to true: by encoding
  scheme in Page~\pageref{encoding}, this means that $v'(y_i^j) - v'(x_i^j)$
  is $1$ if $l_i^j$ is positive and $2$ otherwise. Recall the
  definition of $d_i^j$ as used in
  (\ref{eq:true-implies-constraint-not-satisfied-in-Z'}).

  Finally, from the above discussion and assumptions on $v'$, we get
  that:
  \begin{align*}
    v'(f_i) - v'(y_i^j) & > 2M + d_i^j \cdot (2M) + 2 \quad \text{ if
      $l_i^j$
      is positive} \\
    v'(y_i^j) - v'(e_i) & > 2 + d_i^j\cdot (2M) + 2M \quad \text{ if
      $l_i^j$ is negative}
  \end{align*}
  This contradicts constraint
  (\ref{eq:true-implies-constraint-not-satisfied-in-Z'}) of $Z'$,
  thereby proving that $v' \not \in Z'$.
\end{proof}

\begin{lemma}\label{lem:diag-sat} Let $v \in Z$ be an integer tight
valuation such that for all valuations $v'$ satisfying $v' \regeq_M
v$, we have $v' \not \in Z'$. Then, $\s_v \models \varphi$.
\end{lemma}
\begin{proof}
  Suppose $v \in Z$ is an integer tight valuation such that $\reg{v}$
  does not intersect $Z'$. Let $i\in\{1,\dots,N\}$ and $v_i$ be the
  valuation
  defined below \eqref{eq:border-clocks}. By definition, all tight
  differences $v_i$ are the same as in $v$. Hence $v_i \regeq_M v$,
  and by hypothesis $v_i \not \in Z'$, implying that $v_i$ does not
  satisfy some constraint of $Z'$.

  Valuation $v_i$ satisfies constraints given by
  (\ref{eq:within-a-block-Z'}), (\ref{eq:border-difference-Z'}),
  (\ref{eq:non-border-Z'}) and (\ref{eq:choose-clause-Z'}) by
  construction. The reason that $v_i \not \in Z'$ is therefore due to
  violation of some constraint given by
  (\ref{eq:true-implies-constraint-not-satisfied-in-Z'}).

  For each clause $i' \neq i$, by definition of $v_i$ we have
  $v_i(f_{i'}) - v_i(e_{i'}) = 2M$ and we deduce:
  \begin{align*}
    v_i(f_{i'}) - v_i(y_{i'}^j) & = d_{i'}^j \cdot 2M + 2M +
    v_i(z_{i'}^j) -
    v_i(y_{i'}^j) \quad \text{ if  $l_{i'}^j$ is positive} \\
    v_i(y_{i'}^j) - v_i(e_{i'}) & = v_i(y_{i'}^j) - v_i(x_{i'}^j) +
    d_{i'}^j \cdot 2M + 2M \quad \text{ if $l_{i'}^j$ is negative}
  \end{align*}
  Since $v_i(z_{i'}^j) - v_i(y_{i'}^j) \le 2$ and $v_i(y_{i'}^j) -
  v_i(x_{i'}^j) \le 2$, we get that
  (\ref{eq:true-implies-constraint-not-satisfied-in-Z'}) will be
  satisfied for all clauses $i' \neq i$ and $j\in\{1,2,3\}$.
  Therefore the only possible violation can occur in clause $i$. By
  definition of $v_i$ we have:
  \begin{align*}
    v_i(f_{i}) - v_i(y_{i}^j) & = d_{i}^j \cdot 2M + 2M + 1 +
    v_i(z_{i}^j) -
    v_i(y_{i}^j) \quad \text{ if  $l_{i}^j$ is positive} \\
    v_i(y_{i}^j) - v_i(e_{i}) & = v_i(y_{i}^j) - v_i(x_{i}^j) +
    d_{i}^j \cdot 2M + 2M + 1 \quad \text{ if $l_{i}^j$ is negative}
  \end{align*}
  Since this constraint should be false for some $l_i^j$ we get that
  $v_i(z_i^j) - v_i(y_{i}^j) > 1$ if $l_{i}^j$ is positive and
  $v_i(y_{i}^j) - v_i(x_{i}^j) > 1$ if $l_{i}^j$ is negative. Recall
  that $v$ and $v_i$ are integer tight. Therefore $v_i(y_i^j) -
  v_i(x_i^j) = 1$ when $l_i^j$ is positive and $v_i(y_i^j) -
  v_i(x_i^j) = 2$ if $l_i^j$ is negative.  As valuation $v$ has the
  same value for these differences, by the encoding scheme in Page
  \pageref{encoding}, we get that the literal $l_i^j$
  evaluates to true in $\s_v$. As $i$ was arbitrary, we get that some
  literal in each clause evaluates to true, and hence $\s_v \models
  \varphi$.
\end{proof}

\emph{Proof of Theorem~\ref{thm:hardness}}.
Suppose $\varphi$ is satisfied by assignment $\s$. Pick valuation $v
  \in Z$ that encodes $\s$, that is $\s_v = \s$. Lemma
  \ref{lem:sat-diag} shows that $\reg{v} \cap Z'$ is empty. Hence $Z
  \not \md Z'$.

  Suppose $Z \not \md Z'$. Observe that by the constraints
  that define $Z$, every valuation in $Z$ has the same $\tight^*$
  equivalence classes. Moreover $Z$ is topologically closed. Hence by
  Proposition \ref{prop:topo-closed-closure-inclusion}, there exists
  an integer tight valuation $v \in Z$ such that $\reg{v} \cap Z'$ is
  empty. By Lemma \ref{lem:diag-sat}, assignment $\s_v$ satisfies
  $\varphi$.

  NP-hardness follows since we have given a polynomial time reduction from
  3-SAT.


\end{document}